\documentclass[a4paper,onecolumn,accepted=2025-11-02]{quantumarticle}
\pdfoutput=1

\usepackage[utf8]{inputenc}
\usepackage[english]{babel}
\usepackage[numbers,sort&compress]{natbib}

\usepackage{graphicx}
\usepackage{dcolumn}
\usepackage{bbm}
\usepackage{bm}
\usepackage{qcircuit}
\usepackage{braket}
\usepackage{subfigure}
\usepackage{mathtools}
\usepackage{xparse}
\usepackage{amssymb,amsmath,amsfonts}
\usepackage{dsfont}
\usepackage[dvipsnames]{xcolor}
\usepackage{hyperref}
\hypersetup{
  colorlinks=true,
  citecolor=blue
}
\usepackage[all]{hypcap}  
\usepackage{soul}
\usepackage{amsthm}
\usepackage{tikz}
\usepackage[normalem]{ulem}

\newcommand{\tr}{{\mathrm{Tr}}}

\newcommand{\e}{\mathrm{e}}
\renewcommand\d{\,\mathrm{d}}

\newcommand{\eps}{\epsilon}
\newcommand{\lad}{\lambda}
\def \si {\sigma}

\newcommand{\mc}[1]{\mathcal{#1}}

\def \l {\langle} 
\def \r {\rangle}
\def \q {\quad}

\newcommand{\EE}{\operatorname{\mathbb{E}}}

\newcommand{\cA}{\mathcal{A}}
\newcommand{\cB}{\mathcal{B}}
\newcommand{\cD}{\mathcal{D}}
\newcommand{\cE}{\mathcal{E}}
\newcommand{\cF}{\mathcal{F}}

\newcommand{\cH}{\mathcal{H}}

\newcommand{\cL}{\mathcal{L}}

\newcommand{\cP}{\mathcal{P}}

\newcommand {\Or} {\mathcal{O}}

\definecolor{lw}{RGB}{0,0,255}

\newtheorem{theorem}{Theorem}
\newtheorem{lemma}[theorem]{Lemma}
\newtheorem{corollary}[theorem]{Corollary}
\newtheorem{proposition}[theorem]{Proposition}

\theoremstyle{definition}
\newtheorem{definition}[theorem]{Definition}
\theoremstyle{remark}
\newtheorem{remark}{Remark}
\theoremstyle{definition}
\newtheorem{assumption}{Assumption}
\theoremstyle{definition}

\makeatletter

\begin{document}

\title{A Randomized Method for Simulating Lindblad Equations and Thermal State Preparation}

\author{Hongrui Chen}
\email{hongrui@stanford.edu}
\affiliation{Department of Mathematics, Stanford University, Stanford, CA 94305}
\author{Bowen Li}%
 \email{bowen.li@cityu.edu.hk}
\affiliation{Department of Mathematics, City University of Hong Kong, Kowloon Tong, Hong Kong
SAR}%



\author{Jianfeng Lu}%
 \email{jianfeng@math.duke.edu}
\affiliation{Departments of Mathematics, Physics, and Chemistry, Duke University, Durham, NC 27708}

\author{Lexing Ying} \email{lexing@stanford.edu} \affiliation{Department of Mathematics and ICME, Stanford University, Stanford, CA 94305}%

\begin{abstract}
We study a qDRIFT-type randomized method to simulate Lindblad dynamics by decomposing its generator into an ensemble of Lindbladians, $\mathcal{L} = \sum_{a \in \mathcal{A}} \mathcal{L}_a$, where each $\mathcal{L}_a$ comprises a simple Hamiltonian and a single jump operator. Assuming an efficient quantum simulation is available for the Lindblad evolution $e^{t\mathcal{L}_a}$, we implement $e^{t\mathcal{L}_a}$ for a randomly sampled $\mathcal{L}_a$ at each time step according to a probability distribution $\mu$ over the ensemble $\{\mathcal{L}_a\}_{a \in \mathcal{A}}$. This randomized strategy reduces the quantum cost of simulating Lindblad dynamics, particularly in quantum many-body systems with a large or even infinite number of jump operators. 

Our contributions are two-fold.  First, we provide a detailed convergence analysis of the proposed randomized method, covering both average and typical algorithmic realizations. This analysis extends the known results for the random product formula from closed systems to open systems, ensuring rigorous performance guarantees. Second, based on the random product approximation, we derive a new quantum Gibbs sampler algorithm that utilizes jump operators sampled from a Clifford-random circuit. This generator (i) can be efficiently implemented using our randomized algorithm, and (ii) exhibits a spectral gap lower bound that depends on the spectrum of the Hamiltonian.  Our results present a new instance of a class of Hamiltonians for which the thermal states can be efficiently prepared using a quantum Gibbs sampling algorithm.
\end{abstract}
\section{Introduction}
The open quantum system is a fundamental research area in quantum physics, focusing on the behavior of quantum systems interacting with external environments \cite{breuer2007theory}. In contrast to closed quantum systems, which are isolated and unaffected by their surroundings, open quantum systems experience complex dynamics due to external influences, resulting in phenomena such as decoherence and dissipation. Central to this topic is the Lindblad master equation \cite{cmp/1103899849,cmp/1103860160}, which is the Markovian approximation of a general open quantum dynamics that models the system evolution through dissipative quantum jumps. The Lindblad dynamics has nowadays found fruitful applications in various fields, including materials science \cite{Harbola_2006}, quantum optics \cite{doi:10.1080/00107514.2021.1924281}, and quantum computation \cite{PhysRevLett.75.3788}. In particular, recent developments employed Lindblad dynamics as an algorithmic tool for quantum Gibbs sampling \cite{chen2023quantum,ding2023singleancilla,chen2023efficient,ding2024efficient,gilyén2024quantum}, aiming at preparing the thermal or ground state on a quantum computer. 

Given the wide applications of Lindblad dynamics, developing efficient simulation methods for Lindblad equations is of great importance. Over the past decade, both classical and quantum algorithms have been extensively explored for this purpose. Classical methods \cite{cao2021structure} often face challenges due to exponential scaling in complexity as the system size increases. In contrast, quantum algorithms show promise in reducing this cost from exponential to polynomial scaling. However, many existing quantum methods \cite{childs2016efficient, cleve2016efficient, li2023simulating} face challenges such as requiring a significant number of ancilla qubits and complex control circuits, especially when dealing with Lindblad equations that involve a large number of jump operators.

In this paper, we propose a randomized quantum algorithm for simulating the Lindblad equation. Instead of encoding all jump operators into a quantum circuit, our method implements a single randomly sampled jump operator at each time step. This approach allows for a simpler Lindbladian simulation with the complexity independent of the number of jump operators. 
This simplification is particularly advantageous when dealing with a large or infinite number of jump operators. Our method draws inspiration from the qDRIFT approach for Hamiltonian simulation \cite{Campbell_2019}, which we generalize here from closed to open quantum systems. 

Our approach is particularly useful for quantum Gibbs sampling, which is analogous to classical Markov chain Monte Carlo methods. The Lindblad dynamics-based quantum Gibbs sampler aims to construct an ergodic quantum Markov dynamics that converges to the thermal state. Existing approaches \cite{chen2023quantum,chen2023efficient,ding2024efficient} typically rely on deterministic methods, requiring a large number of jump operators to ensure ergodicity. In contrast, our method utilizes randomly sampled jump operators at each step, leveraging randomness to simplify the implementation. 

As a concrete example, we introduce the \emph{Clifford-random Davies generator}, which consists of an ensemble of exponentially many global jump operators, making deterministic methods impractical. Our randomized approach makes the computation tractable by implementing a single operator sampled from the ensemble at each time step. We further provide a mixing-time analysis for the Clifford-random Davies generator, showing that the spectral gap of the dynamics is closely related to the spectral distribution of the underlying Hamiltonian system. In particular, for some Hamiltonians whose spectra can be predicted by random matrix theory, our analysis demonstrates that the Clifford-random Davies generator achieves efficient thermalization.

\medskip 

\noindent \textbf{General Framework.}
Let $\{\mc{L}_a\}_{a \in \mc{A}}$ be an ensemble of uniformly bounded Lindbladian operators with $\mc{A}$ being a finite or infinite index set, and $\mu$ be a probability measure on $\cA$. We consider the Lindblad master equation with the generator, called the Lindbladian, in the integral form: 
\begin{align} \label{dynamic}
\frac{\d \rho}{\d t} = \bar{\cL} (\rho) := \int_{\cA} \cL_a (\rho) \d \mu(a)\,,
\end{align}
where each Lindbladian $\cL_a$ is given by
\begin{equation} \label{eq:singlelind}
\cL_a (\rho) := \underbrace{-i[H_a,\rho ]}_{\rm coherent} + \underbrace{V_a\rho V_a^{\dag} - \frac{1}{2}\{V_a^\dag V_a,\rho \}}_{\rm dissipative}\,.    
\end{equation}
Here, the coherent and dissipative terms are driven by the Hamiltonian $H_a$ and the jump operator $V_a$, respectively. 

Let $\tau$ be a time step size and $\cL_1,\cL_2,\cdots, \cL_M$ be a sequence of random Lindbladians independently sampled from the ensemble $\{\cL_a:\,a \sim \mu\}$. Then, for a total evolution time $T = M \tau$, the Lindblad dynamics \eqref{dynamic} can be approximated by the random product: \begin{align}\label{randomized}
e^{T \bar{\cL}} \approx e^{\tau {\cL_M}} \cdots  e^{\tau {\cL_2}}e^{\tau {\cL_1}}\,.
 \end{align}
The rigorous convergence analysis will be provided in Section \ref{sec:convergence}. The randomized approximation \eqref{randomized} reduces the task of simulating an ensemble of Lindbladians to {simulating a single Lindbladian $\mc{L}_a$ with simple $H_a$ and $V_a$} over short time intervals. This enables a simpler algorithmic implementation for each time step. For example, we can dilate the Lindbladian into a Hamiltonian acting on the system coupled with an ancilla qubit and then implement Hamiltonian simulation on the joint system \cite{ding2023singleancilla}. 


\subsection{Overview of Main Results}
\subsubsection{Analysis of randomized Lindbladian simulation}
Given an algorithmic realization of the single Lindbladian evolution $e^{\tau \cL_a}$,  denoted by $\cF_{\tau}(\cL_a)$, in practice, we simulate the product \eqref{randomized} by
\begin{align}\label{randomized-alg} \cF_{\tau}(\cL_M) \cF_{\tau}(\cL_{M-1}) \cdots \cF_{\tau}(\cL_1)\,.      
\end{align}
We study the convergence of this randomized algorithm in the following two setups.
\begin{itemize}
\item \emph{The Average Channel.} The mixture of all the random realizations is given by  
\begin{align} \label{average-channel}
\bar{\mathcal{E}}_{\tau, M} := \bigl(\mathbb{E}_{a \sim \mu} \mathcal{F}_{\tau}(\mathcal{L}_a)\bigr)^M\,,
\end{align}
which is a quantum channel capturing the average behavior of the product \eqref{randomized-alg}. 
\item \emph{The Random Channel.}  Without taking the average over the Lindbladian ensemble $\{\mc{L}_a\}_{a \in \mc{A}}$, {each instance} of our randomized method is given by the random channel:
\begin{align} \label{random-channel}
{\mathcal{E}}_{\tau, M} := \mathcal{F}_{\tau}(\mathcal{L}_M) \cdots \mathcal{F}_{\tau}(\mathcal{L}_1)\,,
\end{align}
{where $\mc{L}_i$ is independently sampled from $\{\cL_a:\,a \sim \mu\}$.} 
\end{itemize}

These two setups are motivated by the literature on the qDRFIT method for Hamiltonian simulation  \cite{Campbell_2019,Chen_2021}. In terms of observables, the error bound for the average channel \eqref{average-channel} controls the difference between the expectation values obtained from the randomized simulation algorithm and those of the actual dynamics, i.e., $|\EE \langle O \rangle_{\cE_{\tau,M}(\rho)} - \langle O \rangle_{e^{T\bar{\cL}}(\rho)}|$, {where $\l O \r_\rho: = \tr(O \rho)$ denotes the expectation of an observable $O$ with respect to a state $\rho$.} 
The analysis of typical realizations of the random channel \eqref{random-channel} further provides bounds on the deviation from the expected behavior $\EE |\langle O \rangle_{\cE_{\tau,M}(\rho)}  - \langle O \rangle_{e^{T\bar{\cL}}(\rho)}|$; {See Corollary \ref{coroobs} and Corollary \ref{coro2}.} 

Assuming that the approximation error of the numerical scheme $\cF_{\tau}(\cL_a)$ to $e^{\tau \mc{L}_a}$ is $\Or(\tau^2)$ (see remarks after Assumption \ref{as-alg0}), we provide a non-asymptotic error analysis of the randomized algorithm \eqref{randomized-alg} for simulating Lindblad dynamics in both the average and random setting. 
We first consider the convergence in the average sense.  Let $T= M \tau$ be the total evolution time. Theorem \ref{thm:aver_err} shows that the error $\|e^{T\bar{\cL}} - \bar{\cE}_{\tau, M}\|_{\diamond}$ between the ideal evolution and the averaged implementation \eqref{average-channel}, in terms of the diamond distance \eqref{eq:diamond}, scales as $\Or(T\tau) = \Or(T^2/M)$, matching the error bound of the qDRIFT protocol given in \cite{Campbell_2019} for Hamiltonian simulation.

Then, we study the convergence of the typical realization \eqref{random-channel} with respect to a \emph{weighted $\ell^2$-metric} \eqref{def:weightl2}, which is a variant of the standard quantum $\chi^2$-divergence \cite{temme2010chi}.
We show in Theorem \ref{thm-simulation} that the averaged weighted $\ell^2$-error of simulating \eqref{dynamic} via the random channel \eqref{random-channel} for time $T = M \tau$ generally scales as {$\Or(\sqrt{T\tau}) = \Or(T/\sqrt{M})$}. This result generalizes the analysis of the random product formulas for Hamiltonian systems \cite{huang2020matrix,Chen_2021} to open systems. It is noteworthy that we employ the weighted $\ell^2$-norm here rather than the more commonly used trace norm in the error analysis of quantum simulations.  The reason for this choice will be discussed in Remark \ref{rem2}. 

\subsubsection{Applications to quantum Gibbs sampling}

{In this work, we also examine the behavior of both the average and random quantum channels \eqref{average-channel}--\eqref{random-channel} in the context of quantum Gibbs sampling.} 

\medskip

\paragraph*{General Analysis.} First, we provide a general convergence analysis of applying the randomized Lindbladian simulation framework to the quantum Gibbs sampling task. Let $\sigma$ be the Gibbs state that one aims to prepare. We consider the Lindbladian evolution $\bar{\mathcal{L}}$ given in \eqref{dynamic}, which consists of an ensemble of Hamiltonian and jump operators. Suppose that this dynamics is primitive and preserves $\sigma$ as an invariant state, and has a spectral gap bounded below. Theorem \ref{thermal-convergence} shows that the $\chi^2$-divergence between the output of the average channel ${\bar{\cE}}_{\tau,M}(\rho)$ and the Gibbs state $\sigma$ {decays exponentially, up to an $\Or(\tau^2)$ error term. In contrast}, the $\chi^2$-divergence between the output of the random channel $\cE_{\tau,M}(\rho)$ and the Gibbs state $\sigma$ {decays exponentially with an error of $\Or(\tau)$}. This discussion is similar to the numerical analysis of the classical Langevin dynamic-based sampling algorithm. 

\medskip

\paragraph*{Clifford-random Davies Generator.} 
A notable example is the Clifford-random Davies generator, which employs an ensemble of jump operators with uncorrelated random Hermitian coupling matrices (see Assumption \ref{random-operator}). This construction enables efficient algorithmic implementation while guaranteeing rapid mixing for certain classes of models.

From a computational perspective, the random coupling matrix can be implemented using a unitary design, which admits efficient realization through random Clifford circuits \cite{Dankert2009}. This randomization approach enables the simulation of Lindbladian dynamics with exponentially many jump operators. It is also worth noting that, compared to the implementation of a general Davies generator \cite{chen2023quantum},  this random construction 
allows us to bypass the need for isolating Bohr frequencies by simulating the Heisenberg evolution of coupling operators for an infinitely long period and simultaneous block-encoding techniques. 


{Regarding convergence properties,} the Clifford-random Davies generator induces a classical random walk on the spectral domain, leading to a mixing time bound determined exclusively by the spectral density of the system. For strongly interacting Hamiltonians obeying the semicircle law, the spectral gap of this random Davies generator scales as $\Omega(\beta^{-3/2})$, where $\beta$ denotes the inverse temperature. This behavior is particularly relevant for systems described by random matrix theory, including SYK models of interacting fermions \cite{Kitaev2015,feng2019spectrum} and nuclear energy level statistics \cite{PhysRevLett.52.1}, where such spectral characteristics are observed.


Our result provides the first instance of a fast-converging quantum Gibbs sampling algorithm in the strongly interacting regime, extending previous works that primarily focused on weakly interacting systems \cite{kochanowski2024rapid, rouze2024efficient, bergamaschi2024quantum}. This advancement broadens the literature on efficient thermalization in quantum Gibbs sampling. 

We note that \cite{chen2023fastthermalizationeigenstatethermalization} also provides a fast thermalization analysis for ergodic open quantum systems by reducing the dynamics to a classical random walk on the spectrum. Their analysis is based on the assumption that the coupling operator satisfies the Eigenstate Thermalization Hypothesis (ETH). In contrast, we construct a coupling operator using random Hermitian matrices generated through Clifford random gates, which results in a simpler analysis and enables a concrete algorithmic implementation of the dynamics.

\section{Notations and Preliminaries}

In this section, we shall review the basics of the Lindblad dynamics and the quantum $\chi^2$-divergences. 
We first fix some notations that will be used throughout this work. Let $\mc{H}$ be a finite-dimensional Hilbert space with dimension $N = 2^n$. We denote by $\mc{B}(\mc{H})$ the space of bounded operators and by $\mc{D}(\mc{H}) : =\{\rho \in \mc{B}(\mc{H})\,;\ \rho \succeq 0\,, \tr (\rho) = 1\}$ the set of quantum states. {The subset of full-rank quantum states is then denoted by $\mc{D}_+(\mc{H})$.}
The trace norm and the Hilbert-Schmidt (HS) inner product on $\mc{B}(\mc{H})$ are defined by $\|X\|_1 := \operatorname{Tr}(\sqrt{X^{\dagger} X})$ and $\l X, Y \r := \tr (X^\dag Y)$ for $X, Y \in \mc{B}(\mc{H})$, respectively, where $X^\dag$ is the adjoint of $X$. We use $\|\cdot \|$ to denote the spectral norm and use $\|\cdot\|_F$ to denote the Frobenius norm. {We adopt the standard asymptotic notations $\Or$ and $\Omega$, where $f = \Or(g)$ if $f \le C g$ for some generic constant $C > 0$, which is independent of any parameters we are interested in, and $f=\Omega(g)$ if $g=\Or(f)$. For conciseness, we sometimes write $f \lesssim g$ to mean $f = \Or(g)$ (i.e., $f \leq C g$ for some $C > 0$).} 

 A superoperator $\Phi: \mc{B}(\mc{H}) \to \mc{B}(\mc{H})$ means a bounded linear operator on $\mc{B}(\mc{H})$. 
  Its adjoint with respect to the HS inner product is denoted by $\Phi^\dag$. 
 We also introduce the diamond norm for a superoperator $\Phi$:
\begin{align} \label{eq:diamond}
\|\Phi\|_{\diamond} = \sup_{X \in \cB( \cH \otimes \cH),\,\|X\|_1 = 1 } \|(\Phi \otimes {\rm id}) X \|_1\,,
\end{align}
where ${\rm id}$ is the identity map on $\mc{B}(\mc{H})$. 
We say that a superoperator $\Phi$ on  $\mc{B}(\mc{H})$ is a quantum channel if it is completely positive and trace preserving (CPTP). This means that $\Phi \mc{D}(\mc{H}) \subset \mc{D}(\mc{H})$, ensuring that a quantum channel always maps a quantum state to a quantum state. 

The \emph{Lindblad dynamics} (also called quantum Markov semigroup) is a $C_0$-semigroup of quantum channels $(e^{t\cL})_{t \ge 0} $, which is governed by the Lindblad master equation with the following GKSL form \cite{gorini1976completely,lindblad1976generators}:
\begin{align*}
\frac{\d}{\d t} \rho = \cL(\rho):=  - i [H, \rho] + \sum_{j \in \mc{J}} \left( V_j \rho V_j^\dagger - \frac{1}{2}\left\{V_j^\dag V_j, \rho  \right\} \right)\,.
\end{align*}
Here, $H \in \cB(\cH)$ is the system Hamiltonian, and $V_j \in \cB(\cH)$ are the jump operators characterizing the system-environment interaction.
The generator $\cL$ is usually referred to as the Lindbladian. Noting that for each $t > 0$, $e^{t\cL}$ gives a quantum channel, thus the Lindblad dynamics is contractive under the trace norm: for any quantum states $\rho_1,\rho_2 \in \cD(\cH)$ and time $t > 0$,
\begin{equation*}
  \left\|e^{t\cL}(\rho_1) - e^{t\cL}(\rho_2) \right\|_1 \le \left\|\rho_1 - \rho_2 \right\|_1\,. 
\end{equation*}
 

For a given full-rank state $\si \in \mc{D}(\mc{H})$, define the modular operator: $$\Delta_{\si}(X) = \si X \si^{-1}:\ \mc{B}(\mc{H}) \to \mc{B}(\mc{H})\,,$$
and the weighting operator:
\begin{align} \label{eq:modular}
    \Gamma_\si (X) = \si^{\frac{1}{2}}X\si^{\frac{1}{2}}:\ \mc{B}(\mc{H}) \to \mc{B}(\mc{H})\,.
\end{align} 
The \emph{KMS (Kubo-Martin-Schwinger) inner product} is given by
$$   \langle X, Y \rangle_{\sigma,1/2} = \langle X, \Gamma_{\sigma} (X) \rangle\,,  $$
which induces the KMS norm $\|X\|_{\sigma,1/2} =\sqrt{ \langle X, \Gamma_{\sigma} (X)\rangle}$. 

For self-adjoint $X \in \mc{B}(\mc{H})$, the variance of $X$ with respect to $\sigma \in \mc{D}_+(\mc{H})$ is defined as  \cite[Section 2.2]{muller2018sandwiched}
\begin{align*}
    \mathrm{Var}_\sigma(X): & = \left\| P_I^\perp (X)\right\|_{\sigma,1/2}^2 = \left\|X - \mathrm{Tr}(\sigma X) I\right\|_{\sigma,1/2}^2  \\ & = \|{X}\|_{\sigma,1/2}^2 - \mathrm{Tr}(\sigma X)^2\,,
\end{align*}
where $P_I^{\perp}$ denote the projection onto $\mathrm{span}\{{I}\}^{\perp}$ with respect to the KMS norm. This leads to the definition of \emph{quantum $\chi^2$-divergence} between two quantum states $\rho$ and $\sigma$ \cite{temme2010chi}:
\begin{equation}\label{eq:chidiver}
    \chi^2(\rho,\sigma) = \mathrm{Var}_\sigma(X) =  \|\Gamma_{\sigma}^{-1} (\rho - \sigma)\|_{\sigma,1/2}^2\,, 
\end{equation}
where $X = \Gamma_\si^{-1}(\rho)$ is the relative density of $\rho$ with respect to $\sigma$. 
For convenience, for $Y \in \cB(\cH)$, we use $ \|Y\|_{\sigma,-1/2}$, called the \emph{weighted $\ell^2$-norm}, to denote the KMS norm of $\Gamma_\si^{-1}(Y)$: 
\begin{equation} \label{def:weightl2}
    \|Y\|_{\sigma,-1/2}^2:= \|\Gamma_{\sigma}^{-1}(Y) \|_{\sigma,1/2}^2 = \l Y, \Gamma_\si^{-1}(Y)\r\,,
\end{equation}
which gives 
\begin{equation*}
{\chi^2(\rho,\sigma) = \|\rho - \si\|_{\sigma,-1/2}^2}\,.    
\end{equation*}
The metric $\|\cdot\|_{\sigma,-1/2}$ shall be used for measuring errors in the analysis of the random channel \eqref{random-channel}. 
{Finally, we remark that both the quantum $\chi^2$-divergence and the weighted $\ell^2$-norm can be framed within the broader context of monotone Riemannian metrics in quantum information theory \cite{petz1996geometries,lesniewski1999monotone}. 
Specifically, a Riemannian metric in this setting is a positive definite bilinear form $ M_\sigma(A, B)$ with $\sigma \in \mc{D}_+(\mc{H})$, defined for traceless Hermitian tangent operators $A, B$ in the space $\mc{T}_P$, given by
\begin{align*}
    \mc{T}_P := \{A \in \mc{B}(\mc{H}) \mid A = A^\dagger, \, \tr(A) = 0\}\,.
\end{align*}
The metric \( M_\sigma \) is called \textit{monotone} if, for any quantum channel \(\Phi\), state \(\sigma \in \mc{D}_+(\mc{H})\), and \( A, B \in \mc{T}_P \), it satisfies
\[
M_{\Phi(\sigma)}(\Phi(A), \Phi(B)) \leq M_\sigma(A,B).
\]
As shown in \cite{petz1996geometries,lesniewski1999monotone}, there exists a one-to-one correspondence between the family of monotone Riemannian metrics and a class of convex operator functions. Of particular relevance to our work is the following special one:
\[
M_\sigma(A,B) = \langle A, \Gamma_\sigma^{-1}(B) \rangle, \quad A, B \in \mc{T}_P\,.
\]
It is clear that $\chi^2$-divergence is recovered as $\chi^2(\rho,\sigma) = M_\sigma(\rho - \sigma, \rho - \sigma)$, and, upon relaxing the traceless condition for \( A, B \), $M_\sigma$ yields the weighted $\ell^2$-metric \eqref{def:weightl2}. The monotonicity of \( M_\sigma \) readily implies the contractivity of the weighted norm \(\|\cdot\|_{\sigma,-1/2}\) under any quantum channel \(\Phi\) (cf. \cite[Theorem 4]{temme2010chi}):}
\begin{equation} \label{weighcontra}
{\|\Phi(Y)\|_{\Phi(\sigma),-1/2} \leq \|Y\|_{\sigma,-1/2}.}
\end{equation}

Next, we recall the \emph{quantum detailed balance} condition and \emph{spectral gap} of Lindbladians. A Lindbladian $\mc{L}$ is said to satisfy the KMS detailed balance with respect to a full-rank quantum state $\sigma$ if $\cL^\dag $ is self-adjoint under the KMS inner product, equivalently, $\cL \Gamma_{\sigma} = \Gamma_{\sigma} \cL^\dag$. The KMS detailed balance condition ensures that $\sigma$ is an invariant state of the dynamics $e^{t\cL}$. For a detailed balanced $\cL$, the spectral gap is defined by the smallest nonzero eigenvalue of $-\cL$, which is a fundamental tool to characterize the convergence of the dynamics $e^{t\cL}$ under the quantum $\chi^2$-divergence. A Lindblad dynamics $e^{t \mc{L}}$ is called \emph{primitive} if it admits a unique full-rank invariant state. 

\begin{definition}\label{Spectral gap}
The spectral gap of a primitive detailed balanced Lindbladian $\cL$ for $\sigma \in \mc{D}_+(\mc{H})$ is given by
\begin{equation*}
     \lad(\mc{L}) = \inf_{X \in \mc{B}(\mc{H})\backslash\{0\}} -\frac{\langle X, \cL^\dag  X\rangle_{\sigma,1/2}}{\mathrm{Var}_\si(X)}\,.
\end{equation*}
\end{definition}

\begin{proposition}[{\cite[Lemma 12]{temme2010chi}}] \label{prop:gap}
Suppose the symmetrized Lindbladian $\frac{1}{2}(\cL^\dag + \cL^{\mathrm{KMS}})$ is primitive with spectral gap $\eta > 0$, where $\cL^{\mathrm{KMS}}$ denotes the adjoint of $\cL^\dag$ with respect to the KMS inner product. Then, for any quantum state $\rho$ and time $t \geq 0$, the $\chi^2$-divergence decays exponentially as
\begin{align*}
    \chi^2(e^{t \cL}(\rho), \sigma) \leq e^{-2\eta t} \, \chi^2(\rho, \sigma).
\end{align*}
\end{proposition}

We refer to \cite{temme2010chi,Carlen_2019} for a detailed discussion on the detailed balance condition and the spectral gap.
\section{Convergence Analysis for Lindbladian Simulation} \label{sec:convergence}
In this section, we analyze the error of the proposed randomized algorithm \eqref{randomized-alg} for approximating the exact Lindblad evolution $e^{t \bar{\cL}}$ {over a total time horizon $T = M\tau$, where $M$ is the number of iterations and $\tau$ is the step size.}

\subsection{Analysis of the Average Channel} \label{section-3.1}
We begin by analyzing the approximation error of the average channel $\bar{\cE}_{\tau,M}$ defined in \eqref{average-channel} in the diamond norm $\|\cdot\|_\diamond$. Our analysis requires the following key assumption: the spectral norms of both the Hamiltonians $\{H_a\}$ and jump operators $\{V_{a}\}$ appearing in the Lindbladians $\{\cL_a\}_{a\in\cA}$ are uniformly bounded. 

\begin{assumption} \label{as0}
There exists $\lad > 0$ such that for any $a \in \cA$, there holds 
$$     \|H_a \| \le \lambda\,, \quad \|V_a\|^2 \le \lambda\,,          $$
for the Hamiltonian and jumps involved in $\cL_a$ \eqref{eq:singlelind}.
\end{assumption}
The norm bound in Assumption \ref{as0} immediately implies the diamond norm bound for the Lindbladians.

\begin{lemma} \label{op-bound0}
Under Assumption \ref{as0}, the diamond norm of $\{\cL_a\}_{a \in \mc{A}}$ is uniformly bounded by
$$  \|\cL_a \|_{\diamond} \lesssim \lambda\,.   $$
\end{lemma}

We refer to \cite[Section 4]{cleve2016efficient} for the proof of Lemma \ref{op-bound0}.  This bound is crucial for the analysis of the approximation error of the random product approximation \eqref{randomized}. We also make the following assumption that the local truncation error of implementing a single Lindbladian is of order $\Or(\tau^2)$.

\begin{assumption} \label{as-alg0}
Let $\{\mc{L}_a\}_{a \in \mc{A}}$ satisfy Assumption \ref{as0}. For the quantum simulation $\mc{F}_\tau(\mc{L}_a)$ of a Lindblad evolution $e^{\tau \mc{L}_a}$ up to a small time $\tau = \Or(\lad^{-1})$, we assume  
\begin{align} \label{local-truncation} 
\|\cF_\tau(\cL_a) - e^{\tau \cL_a}  \|_{\diamond} \lesssim \lambda^2    \tau^2\,.
\end{align}
\end{assumption}

We will demonstrate in Appendix \ref{app-B} that the condition \eqref{local-truncation} holds for the Hamiltonian-based algorithm presented by \cite{cleve2016efficient,ding2023singleancilla}. It is noteworthy that higher-order methods $\mc{F}_\tau$ are not helpful here. Even if the single-time-step evolution $e^{\tau \cL_a}$ is implemented exactly, the qDRIFT approximation \eqref{randomized} still results in an $\Or(\tau^2)$ local truncation error, {since the product \eqref{randomized}  is essentially a randomized compilation of the first-order Lie-Trotter decomposition.} Consequently, a first-order quantum simulation for $e^{\tau \cL_a}$, as proposed in \cite{cleve2016efficient, ding2023singleancilla}, is sufficient for our purposes. Additional examples of first-order methods for $e^{\tau \cL_a}$ can be found in \cite{cao2021structure, ding2024simulating}.


Under the above assumptions, we establish an upper bound in the diamond norm for the approximation error between the quantum channel $\bar{\cE}_{\tau,M}$ and the ideal evolution $e^{T\bar{\cL}}$. 

\begin{theorem} \label{thm:aver_err}
Under Assumptions \ref{as0} and \ref{as-alg0}, let $T = \tau M$ be the total simulation time with $\tau = \Or(\lambda^{-1})$. The diamond norm distance between the average channel \eqref{average-channel} and the exact evolution \eqref{dynamic} is bounded by
\begin{align*}
\|\bar{\cE}_{\tau,M}- e^{T\bar{\cL}} \|_{\diamond} \lesssim \frac{\lambda^2 T^2}{M}\,.
\end{align*}
\end{theorem}

The proof is given in Appendix \ref{sec:prof1}. We outline the proof idea below. From Assumption \ref{as-alg0}, the local truncation error of the average channel $\bar{\cE}_{\tau,M}$ is $\Or(\lambda^2 \tau^2)$. Combining this with the contractivity of the quantum channel, we know that the error accumulates linearly over time. Thus, we obtain an $M \Or(\lambda^2 \tau^2) = \Or(\lambda^2 T\tau) = \Or(\lambda^2 T^2/M)$ error bound. As an immediate corollary of Theorem \ref{thm:aver_err}, we obtain the following error bound for observable estimation.

\begin{corollary} \label{coroobs}
Under Assumptions \ref{as0} and \ref{as-alg0}, let $T = \tau M$ be the total simulation time. Let $\rho$ be an initial state, for any observable $O \in \mc{B}(\mc{H})$, we have, for $\tau = \Or(\lambda^{-1})$, 
\begin{equation*}
    \big|\EE \langle O \rangle_{\cE_{\tau,M}(\rho)} - \langle O \rangle_{e^{T\bar{\cL}}(\rho)} \big| \lesssim \|O\|\frac{\lambda^2 T^2}{M}\,.
\end{equation*}
\end{corollary}

\begin{remark}[Comparison with existing and concurrent works] \label{rem:compare}
{The Trotter-Suzuki product-formula algorithm is widely recognized as a simple yet efficient method for simulating Hamiltonian dynamics on a quantum computer; see \cite{Childs_2021} and references therein for recent developments. For simplicity, let $\cA$ be a discrete ensemble with a uniform distribution $\mu$. Then, the dynamics in \eqref{dynamic} reduce to the standard case with a finite number of jump operators:  
\begin{align}  \label{finite-jump}
\bar{\cL} =\frac{1}{|\cA|} \sum_{a\in \cA}\cL_{a}\,,
\end{align}  
where we assume that $\mc{L}_a$ is sufficiently simple, in the sense that simulating $e^{\tau \mc{L}_a}$ over a short time $\tau$ has a gate complexity of $\Or(1)$. Kliesch et al. \cite{Kliesch_2011} introduced the first quantum algorithm for simulating Lindblad dynamics, based on the first-order Lie-Trotter formula, with a gate complexity of $\Or\left(|\cA| \lambda^2 T^2 /\epsilon \right)$ for \eqref{finite-jump} under Assumption \ref{as0}. In contrast, our randomized approach reduces the complexity to $\Or\left(\lambda^2 T^2/\epsilon \right)$, making it independent of the number of jump operators. Thus, when $|\cA|$ is large, as is often the case in applications in thermal state preparation or electronic structure
Hamiltonians, our method could substantially reduce the simulation cost.

The scaling in simulation time $T$ and accuracy $\eps$ can be further improved using the second-order Suzuki formula, which achieves a gate complexity of $\Or(|\cA| \sqrt{\lambda^3 T^3/\eps})$, as shown in \cite{werner2016positive}. However, unlike Hamiltonian simulation, higher-order product formulas cannot be directly applied to Lindblad dynamics \cite{childs2016efficient}. This limitation arises because, unlike unitary dynamics, Markovian quantum dynamics cannot evolve backward in time while obtaining a CPTP map. More recent works \cite{childs2016efficient,cleve2016efficient,li2023simulating,ding2024simulating} have significantly improved simulation costs using techniques such as Stinespring dilation, oblivious amplitude amplification, and Dyson series expansions. These methods could achieve costs that scale linearly with $T$ and polylogarithmically with $\eps$, but they still depend at least linearly on $|\cA|$.  

Although product-formula-based algorithms do not achieve the optimal scaling with respect to $T$ and $\eps$, they are more practical in terms of implementation, particularly their lower-order forms, compared to the aforementioned state-of-the-art quantum algorithms for Lindbladian simulation, which are highly technical and often rely on complex control-logic operations. Our work enhances the efficiency of first-order product formulas for Lindbladian simulation, which remains valuable in the era of early fault-tolerant quantum computing.  

After completing the first manuscript, we became aware of concurrent works aimed at improving lower-order product-formula-based Lindbladian simulation \cite{borras2025quantum,peng2024quantum,david2024faster}. Specifically, Borras and Marvian \cite{borras2025quantum} proposed a quantum algorithm that combines randomized compiling with the second-order product formula, achieving a gate complexity of $\Or(\sqrt{\lambda^3 T^3/\eps})$. However, their method is limited to a special class of Lindbladians where the short-time evolution of the dissipative part is assumed to be approximately a convex combination of unitary channels with a fixed state.  
The work by Peng et al. \cite{peng2024quantum} introduced a new approximation channel for short-time Lindblad evolution, motivated by the classical quantum trajectory method. They designed two quantum algorithms with reduced dependence on the number of jump operators. Notably, their first algorithm \cite[Theorem 1]{peng2024quantum} achieves the same scaling as ours, but with an additional dependence on the system size $n$.  
Moreover, David et al. \cite{david2024faster} extended the randomized approach from Hamiltonian simulation \cite{childs2019faster} to Lindblad dynamics, proposing first- and second-order randomized Trotter-Suzuki formulas that achieve gate complexities of $\Or(|\mc{A}|\sqrt{\lambda^3 T^3/\eps})$ and $\Or(\sqrt{|\mc{A}|\lambda^3 T^3/\eps})$, respectively. The key idea is to randomize the order of the summands in the product formula, which is qualitatively different from the qDRIFT method. Additionally, in \cite[Section 5]{david2024faster}, a qDRIFT-type algorithm is proposed with the same error bound as ours (Theorem \ref{thm:aver_err}). However, in addition to this result, we also establish the error for each individual realization (see Section \ref{random-realization}).  
Beyond these algorithmic advances, another main contribution of our work lies in applying the proposed randomized framework to the quantum thermal state preparation (see Section \ref{sec:appgibbs}).}
\end{remark}



\subsection{Analysis of the Random Channel} \label{random-realization}
In this section, we will show that without taking the average of the randomized channel, each realization of the algorithm, modeled by the random channel $\cE_{\tau,M} = \cF_\tau(\cL_M) \cF_{\tau}(\cL_{M-1})\cdots \cF_{\tau}(\cL_1)$ in \eqref{random-channel}, will typically give a small numerical error for the quantum simulation task. Suppose that $\sigma \in \mc{D}_+(\cH)$ is a full-rank invariant state of \eqref{dynamic}, i.e., $\bar{\cL}(\sigma) = 0$. We consider the error under the weighted $\ell^2$-norm \eqref{def:weightl2}:
 \begin{equation} \label{eq:l2error}
 \begin{aligned}
 e_{\tau, M}(\rho) :=\big\|{\mathcal{E}}_{\tau, M}(\rho) - e^{T \bar{\mathcal{L}}} (\rho)\big\|_{\sigma,-1 / 2}\,.
 \end{aligned}
 \end{equation}
 For our analysis, we require several assumptions. First, we impose a \emph{weighted} operator norm condition on both the system Hamiltonian and jump operators comprising the Lindbladian, as the weighted $\ell^2$-norm is employed in the error analysis. This differs subtly from the standard operator norm bound in Assumption~\ref{as0} (Section~\ref{section-3.1}).

\begin{assumption} \label{as1}
{Given a full-rank invariant state
$\sigma$ of $e^{t \bar{\mc{L}}}$}, there exists $\Lambda > 0$ such that for any $a \in \cA$, 
\begin{align*}
\| \Delta_\sigma^{1/4} H_a   \| \le \Lambda\,,\quad    \| \Delta_{\sigma}^{1/4} V_a \|^2 \le \Lambda\,,\quad \| \Delta_{\sigma}^{1/4} V_a^\dag  \|^2 \le \Lambda\,,
\end{align*}
for the Hamiltonian and jump involved in $\cL_a$ \eqref{eq:singlelind}. Here $\|\cdot \|$ is the standard spectral norm.  
\end{assumption}

Similar to Lemma~\ref{op-bound0}, Assumption~\ref{as1} yields an operator norm bound $\|\mathcal{L}_a\|_{(\sigma,-1/2) \to (\sigma,-1/2)}$ for {Lindbladians $\mathcal{L}_a$}, with the proof provided in Appendix~\ref{sec:app1}.

\begin{lemma}\label{op-bound1}
Under Assumption \ref{as1}, for any $a \in \cA$, we have 
$$  \|\cL_a \|_{(\sigma,-1/2) \to (\sigma,-1/2)}  \lesssim \Lambda\,.   $$
\end{lemma}
We also assume a local truncation error bound for the algorithm $\cF_{\tau}$ under the $\|\cdot\|_{\sigma,-1/2}$ norm.
\begin{assumption}\label{as-alg1}
{Given $\sigma \in \mc{D}_+(\cH)$ as in Assumption \ref{as1}},
the algorithm $\cF_{\tau}$ for implementing a single-step Lindblad dynamics satisfies
\begin{align}
\left\|e^{\tau \cL_a}(\rho) - \cF_{\tau}(\cL_a)(\rho)  \right\|_{\sigma,-1/2} \lesssim \Lambda^2 \tau^2 \|\rho\|_{\sigma,-1/2}\,,
\end{align}
uniformly in $\rho \in \cD(\cH)$ and $a\in \cA$. 
\end{assumption}

The truncation error condition in Assumption \ref{as-alg1} applies to general first-order open quantum system simulations, analogous to the diamond-norm condition in Assumption \ref{as-alg0}. This is because the local truncation error primarily depends on the order of the Taylor expansion. As a concrete example, we verify in Appendix \ref{app-B} that the Hamiltonian simulation-based algorithm from \cite{cleve2016efficient,ding2023singleancilla} satisfies Assumption \ref{as-alg1}.

 
\begin{theorem}  \label{thm-simulation}
Suppose Assumptions \ref{as1} and \ref{as-alg1} hold. Let $T = \tau M$ be the total simulation time with {$\tau = \Or(\min\{\frac{1}{\Lambda},\frac{1}{\Lambda^2 T}\})$}. 
For any initial state $\rho$, the average $\|\cdot\|_{\sigma,-1/2}$-error \eqref{eq:l2error} between the random channel \eqref{random-channel} and the exact evolution \eqref{dynamic} is bounded as follows:
\begin{equation} \label{eq:error1}
  {\EE [e_{\tau,M}(\rho)]   \lesssim   \frac{\Lambda T}{\sqrt{M}}\sqrt{1+\chi^2(\rho,\sigma)}\,.}   
\end{equation}
\end{theorem}

For general Lindblad dynamics, we established a half-order convergence rate {of $\Or(\Lambda \sqrt{T \tau})$ in Theorem \ref{thm-simulation}. If we further assume that the dynamics $e^{t \bar{\cL}}$ is primitive and admits a spectral gap (see Theorem \ref{as2} below), the estimate \eqref{eq:error1} can be refined with an improved prefactor for the initial $\chi^2$-divergence $\chi^2(\rho\,,\sigma)$ term:
\begin{align} \label{improvest}
\mathbb{E} [e_{\tau,M}(\rho)] \lesssim \frac{\Lambda T}{\sqrt{M}} + \frac{\Lambda T}{M} \sqrt{\frac{1 - e^{-2 \eta T}}{1 - e^{- 2 \eta \tau}} \chi^2(\rho,\sigma)}\,.
\end{align}
The proofs of Theorem \ref{thm-simulation} and the refined estimate \eqref{improvest} are provided in Appendix~\ref{proof-thm-simulation}.} 
The key insight of the proof lies in leveraging the independence of random errors across time steps. Intuitively, for a single step, the local truncation error $e^{\tau \bar{\cL}} - e^{\tau {\cL_a}}$ is of order $\Or(\tau)$. The total error behaves as an independent sum of $\Or(\tau)$ errors over $M$ steps, scaling as $\Or(\sqrt{M}\tau) = \Or(M^{-1/2})$, as suggested by the central limit theorem (see also Remark \ref{rem2}). Similarly to Corollary \ref{coroobs}, we can also derive an error bound for observable estimation.


\begin{corollary} \label{coro2}
Under Assumptions \ref{as1} and \ref{as-alg1}, let $T = \tau M$ be the total simulation time with {$\tau = \Or(\min\{\frac{1}{\Lambda},\frac{1}{\Lambda^2 T}\})$}, and $\cE_{\tau,M}$ be the random channel given in \eqref{random-channel}. 
For any observable $O \in \mc{B}(\mc{H})$, we have
$$ 
    {\EE \big|\langle O \rangle_{\cE_{\tau,M}(\rho)}  - \langle O \rangle_{e^{T\bar{\cL}}(\rho)} \big| \lesssim \|O\|_{\sigma,1/2} \frac{\Lambda T}{\sqrt{M}}\sqrt{1+\chi^2(\rho,\sigma)}\,.} 
$$
\end{corollary}

\begin{proof}
{It suffices to note 
\begin{align*}
    \langle O \rangle_{\cE_{\tau,M}(\rho)}  - \langle O \rangle_{e^{T\bar{\cL}}(\rho)} & = \tr \left( O \big(\cE_{\tau,M}(\rho) - e^{T\bar{\cL}}(\rho) \big) \right) \\
    & = \left\l O, \Gamma_\si^{-1}\big(\cE_{\tau,M}(\rho) - e^{T\bar{\cL}}(\rho) \big) \right\r_{\si,1/2} \\
    & \le \| O\|_{\sigma,1/2} \big\| \cE_{\tau,M}(\rho) - e^{T\bar{\cL}}(\rho) \big\|_{\si,-1/2}\,.
\end{align*}
The proof is complete by \eqref{eq:error1}.}
\end{proof}

\begin{remark}[Discussion on the choice of metric] \label{rem2}
The intuition that the sum of $M$ independent terms scales as $\Or(\sqrt{M})$ requires the underlying norm to define a \emph{type-2 space} \footnote{
A Banach space $(X, \|\cdot\|)$ is called a \emph{type-2 space} if there exists a constant $C>0$ such that for any finite set $\{f_i\}_{i=1}^n \subset X$,
\[
\left(\mathbb{E}\left\|\sum_{i=1}^n \varepsilon_i f_i\right\|^2\right)^{1/2} \leq C\left(\sum_{i=1}^n\|f_i\|^2\right)^{1/2},
\]
where $\{\varepsilon_i\}$ are independent Rademacher random variables with $\mathbb{P}(\varepsilon_i = \pm1) = \frac{1}{2}$.
}. {This property serves as a probabilistic generalization of the Pythagorean identity for Hilbert spaces and provides a prototype for the more general concepts of uniform smoothness and subquadratic averages in Banach space norms for the concentration analysis; see \cite{ledoux2013probability,huang2020matrix} for further details. It is worth noting that any Hilbert space, and more generally the space of $p$-th Schatten-class operators ($p \ge 2$), qualifies as a type-2 space.} However, the matrix space with the trace norm does not fall into this category, meaning we cannot expect a similar $\mathcal{O}(\sqrt{M})$ growth in this case. 
As a remark, Chen et al. \cite{Chen_2021} analyzed concentration for the qDRIFT random product formula for Hamiltonian simulation using the trace norm. {However, their proof essentially relies on the spectral norm for unitary matrices or the $\ell^2$-norm for pure states (vectors).} For the analysis of open quantum systems, one must consider mixed quantum states, where the \(\ell^2\)-norm (i.e., the Frobenius norm) becomes problematic. This is because Lindblad dynamics are not contractive under this norm, leading to the potential accumulation of simulation error at an exponential rate over time \(T\); see \cite[Section D.2]{ding2023singleancilla}.

{Our analysis instead employs the weighted $\ell^2$-norm $\|\cdot\|_{\sigma,-1/2}$, a variant of the monotone Riemannian metric \cite{lesniewski1999monotone}.} This norm is both type-2, {thanks to the Hilbert space structure,} and contractive under quantum channels, making it suited for our purposes. {Theorem \ref{thm-simulation} ensures that, for a \emph{fixed but arbitrary} initial state \(\rho\), the relative simulation error \(\EE [e_{\tau,M}(\rho)]/\|\rho\|_{\sigma,-1/2}\) is small. In other words, the simulation error is small compared to the weighted $\ell^2$-norm of the given initial state. In the worst-case scenario, the prefactor $\chi^2(\rho,\sigma)$ in the bound \eqref{eq:error1} can scale exponentially with the number of qubits $n$, by noting
\begin{align*}
\max_{\rho \in \mc{D}(\mc{H})} \chi^2(\rho,\sigma) = \sigma_{\min}^{-1} - 1 = \Omega(2^n)\,.
\end{align*}
Here $\sigma_{\min}$ is the minimal eigenvalue of $\sigma$, bounded above by $1 / \dim(\mc{H}) = 2^{-n}$.}
\end{remark}

\section{Applications in Quantum Gibbs Sampling} \label{sec:appgibbs}

In this section, we present a randomized algorithm for quantum Gibbs sampling and its analysis.  
Let $\sigma \propto e^{-\beta H}$ be the thermal state we aim to prepare, where $H \in \mathbb{C}^{N \times N}$ is a quantum many-body Hamiltonian with $N = 2^n$. 

\subsection{General Analysis}

We begin by analyzing the general case where an ensemble of jump operators generates dynamics \eqref{dynamic} with average generator $\bar{\cL}$ that converges to the thermal state $\sigma$. Under this framework, we establish $\chi^2$-divergence convergence bounds for both the deterministic average channel $\bar{\cE}_{\tau, M}$ in \eqref{average-channel} and the randomized implementation $\cE_{\tau, M}$ in \eqref{random-channel}. The analysis requires the primitivity of the dynamics, specifically, that the spectral gap of the symmetrized generator is bounded below, which we formalize in the Assumption \ref{as2} below. 


\begin{assumption} \label{as2}
The Lindbladian $\bar{\cL}$ is primitive {with the unique full-rank invariant state $\sigma \in \mc{D}_+(\mc{H})$.} Moreover, the spectral gap of its symmetrized generator {$(\bar{\cL}^\dag +  \bar{\cL}^{{\rm KMS}})/2$} is bounded below by $\eta > 0$. 
\end{assumption}

\begin{theorem} \label{thermal-convergence}
Under Assumptions \ref{as1},\ref{as-alg1},\ref{as2}, for a step size {$\tau = \Or\big(\min\{\frac{1}{\eta}, \frac{1}{\Lambda}, \frac{1}{T \Lambda^2}\}\big)$}, and any initial state $\rho \in \mc{D}(\cH)$,
\begin{itemize}
\item The average channel $\bar{\cE}_{\tau, M}$ in \eqref{average-channel} satisfies
\begin{align} \label{convergence-average-sampler}
     \chi^2\left({\bar{\mathcal{E}}}_{\tau, M}(\rho), \sigma\right) \lesssim e^{- 2 \eta \tau M} \chi^2({\rho},\sigma) + 
\frac{\tau^2 \Lambda^4}{\eta^2}\,.
\end{align}
\item The randomized channel ${{\cE}}_{\tau,M}$ in \eqref{random-channel} satisfies
\begin{align} \label{convergence-sampler}
 \EE\big[ {\chi^2(\cE_{\tau,M}(\rho),\sigma)} \big]  \lesssim \e^{-2\eta \tau M} \chi^2(\rho,\sigma) + \frac{\tau \Lambda^2}{\eta}\,.
\end{align}
\end{itemize}
\end{theorem}


The proof is deferred to Appendix \ref{proof-thermal-convergence}. {Note that the spectral gap condition implies the exponential $\chi^2$-convergence $e^{t\bar{\cL}}(\rho)$ to $\sigma$ with rate $2 \eta > 0$, as stated in Proposition \ref{prop:gap}. In Theorem \ref{thermal-convergence}, we show that our quantum simulation preserves this exponential convergence, but introduces 
an $\mathcal{O}(\tau^2)$ (resp., $\mathcal{O}(\tau)$) additive error for the average channel $\bar{\cE}_{\tau,M}$ (resp., the randomized one $\cE_{\tau,M}$). Therefore, under the infinite-time evolution, the algorithm would converge to a state $\sigma_*$ close to the thermal state $\sigma \propto e^{-\beta H}$ with an $\mathcal{O}(\tau)$ error. The residual error arises because each implementation step $\mc{F}_\tau(\mc{L}_a)$ only approximately preserves the thermal state $\sigma$, rather than exact preservation.} This parallels the classical Langevin Monte Carlo algorithm with Euler-Maruyama discretization \cite{chewi2023log}, {where the invariant measure is similarly perturbed by the step size.}

\subsection{The Clifford-random Davies Generator}
In this section, we introduce a randomized construction of a Davies generator based on random Clifford circuits, showcasing the efficacy of our method in realizing an efficient quantum Gibbs sampler.

The Davies generator yields a class of quantum Markovian systems that thermalize to the Gibbs state $\sigma$. Let $\{A_a\}_{a \in \cA}$ be a set of coupling operators satisfying $\{A_a\}_{a \in \cA} = \{A_a^\dag\}_{a \in \cA}$. The Davies generator is described by 
\begin{align} \label{Davies-generator}
\begin{aligned}
\mathcal{L}_D(\rho)  = \sum_{a \in \mathcal{A}} \sum_{\omega \in B_H} \gamma(\omega)\left(A_a(\omega)\rho A_a(\omega)^{\dagger} -  \frac{1}{2}\left\{A_a(\omega)^{\dagger} A_a(\omega), \rho\right\}\right)\,,
\end{aligned}
\end{align}
where $\omega \in B_H:=\operatorname{spec}(H)-\operatorname{spec}(H)$ are the Bohr frequencies, the set of energy differences of the Hamiltonian. The weight function $\gamma(\cdot)$ is non-negative and satisfies the KMS condition 
\begin{align} \label{eq:kmscondition}
    \gamma(\omega)=e^{-\beta \omega} \gamma(-\omega)\,.    
\end{align}
For each Bohr frequency $\omega \in B_H$, the operator $A_a(\omega)$ in \eqref{Davies-generator} is defined as a restriction of $A_a$ onto the eigenstates with $\omega$ energy difference:
\begin{align} \label{def:aomega}
    A_a(\omega)=\sum_{\lambda_i-\lambda_j=\omega} \left|\psi_i\right\rangle\left\langle\psi_i\left|A_a\right| \psi_j\right\rangle\left\langle\psi_j\right|,
\end{align}
where $\lad_i$ are eigenvalues of $H$,  counting multiplicity, and $\ket{\psi_i}$ is the associated eigenfunction. 

We consider the specific case that the coupling operators $\{A_a\}_{a \in \cA}$, {satisfying $\{A_a\}_{a \in \cA} = \{A_a^\dag\}_{a \in \cA}$}, in the Davies generator are given by an ensemble of random Hermitian operators without correlations. In particular, we assume the random coupling operators satisfy the following condition:



\begin{assumption} \label{random-operator}
 Define $A_{i,j} = \bra{\psi_i}A\ket{\psi_j}$, and assume $$\EE A_{i,j}\overline{A_{l,k}} = {\xi^2} \delta_{i,l}\delta_{j,k}\,,$$ where  $\xi = \Theta(N^{-1/2})$ is a positive constant.
\end{assumption}


As shown in the following lemma (with proof given in Appendix \ref{app:random}) and Section \ref{sec:imple}, this type of random operator can be constructed using random Clifford circuits, so we refer to the generator as the Clifford-random Davies generator.  

\begin{lemma} \label{lemma-random-design}
The ensemble of coupling operators satisfying Assumption \ref{random-operator} can be constructed as follows. Let $D$ be a {diagonal random operator defined by} 
$D = \sigma_z^{p_1} \otimes \sigma_z^{p_2} \otimes \cdots \otimes \sigma_z^{p_n}$, where $p_1,\cdots,p_n$ are independent random binary variable. This ensures that
$$\EE \left[D_{ii}D_{jj}\right] = \delta_{ij}\,.$$
Let $U_2 \in \cB(\cH)$ be a random unitary sampled from the unitary 2-design, i.e., for any $O \in \cB( \cH \otimes \cH )$, 
$$ \EE \left[U_2^{\otimes 2} O U_2^{\dag \otimes 2} \right] = \EE_{U \sim \mathrm{Haar}} \left[U^{\otimes 2} O U^{\dag \otimes 2}\right]\,.$$
Then $A = U_2 D U_2^\dag$ is Hermitian and satisfies Assumption \ref{random-operator}.    
\end{lemma}

\subsubsection{Algorithmic Implementation of the Clifford-random Davies Generator} \label{sec:imple}

Before investigating the convergence behavior of the Clifford-random Davies Generator, we first discuss how to implement it on a quantum computer. Our approach utilizes the integral form representation of the jump operator given in \cite{ding2023singleancilla, ding2024efficient},  which is then implemented using our randomized algorithm.


Consider the Lindblad equation
\begin{align}\label{lindblad}
\frac{\d }{\d t} \rho = \bar{\cL}(\rho) = \EE [\cL_{K_a}(\rho)]\,,\quad
\cL_{K_a}[\rho] = K_a \rho K_a^\dag - \frac{1}{2}\{K_a^\dag K_a ,\rho \} \,, 
\end{align}
with $K_a$ given by the following integral form:
\begin{align} \label{def:kintegral}
K_a: =\int_{-\infty}^\infty f(s)e^{iHs}A_ae^{-iHs} \d s   = \sum_{i,j \in [N]} \hat{f}(\lambda_i - \lambda_j)\ket{\psi_i}\bra{\psi_i} A_a \ket{\psi_j} \bra{\psi_j}, 
\end{align}
where $\{A_a\}$ is an ensemble of Hermitian coupling operators satisfying Assumption \ref{random-operator} and $f$ is a weight function such that {its Fourier transform is given by}
\begin{align} \label{weight-def}
    \hat{f}(\omega) = \sqrt{\gamma(\omega)} \quad \text{for any}\  \omega \in B_H \subset \left[-\|H\|,\|H\|\right],
\end{align}
with $\gamma(\omega)$ satisfying \eqref{eq:kmscondition}. Then, the construction in \eqref{lindblad}, with the expectation taken over random coupling operators satisfying Assumption \ref{random-operator}, yields the Davies generator.

\begin{theorem} \label{thm-construction-Davies}
Under Assumption \ref{random-operator}, the generator of the Lindblad equation \eqref{lindblad} can be written as
 \begin{align*}
     \EE_a[\cL_{K_a}(\rho)] = \sum_{\omega \in B_H} \gamma(\omega) \EE_a\left( A_a(\omega) \rho A_a(\omega)^\dag - \frac{1}{2} \left\{A_a(\omega)^\dag A_a(\omega), \rho \right\}\right),
 \end{align*}
 which is a Davies generator. 
\end{theorem}
The proof is given in Appendix \ref{proof-thm-construction-Davies}.
To simulate the dynamic \eqref{lindblad}, we employ our randomized algorithm to reduce the problem to simulating $e^{\tau \cL_K}$ for randomly sampled $K$ at each time step. Applying the Hamiltonian simulation-based algorithm (See Appendix \ref{app-B} for details), each step $e^{\tau \cL_K}$  involves implementing the Hamiltonian simulation $e^{-i\sqrt{\tau} \tilde{K}}$, where $\tilde{K}$ is given by dilation of $K$ on the system coupled with an ancilla qubit:
$$ \tilde{K} = \left(\begin{array}{cc}
0 & K^{\dagger} \\
K & 0
\end{array}\right)$$
We outline the key idea of implementing $e^{-i\sqrt{\tau} \tilde{K}}$ here and refer to \cite[Section III]{ding2023singleancilla} for further implementation details and numerical analysis. First, we select $f$ as a sufficiently smooth and rapidly decaying function that satisfies \eqref{weight-def}. This enables us to truncate and discretize the integration, leading to:
$$ K \approx \sum_{l=0}^{M} f\left(s_l\right) e^{i H s_l} A e^{-i H s_l} w_l\,, $$
where $s_l$ are the discretization points and $w_l$ are some weights. Then, we obtain an approximation of the Hamiltonian $\tilde{K}$ as follows:
\begin{align} \label{K-decompose}
\tilde{K} \approx \sum_{l=0}^M \tilde{H}_l,
\end{align}
where $\tilde{H}_l$ is defined by
\begin{align*}
 \widetilde{H}_l & =\left(\begin{array}{cc}
0 & \overline{f\left(s_l\right)} e^{is_lH}Ae^{-is_lH} w_l \\
f\left(s_l\right) e^{is_lH}Ae^{-is_lH} w_l & 0
\end{array}\right)  
\end{align*}
Next, by Trotterizing the summation \eqref{K-decompose}, we only need to simulate each $\tilde{H}_l$ separately. Note that 
$\tilde{H}_l$ has a factorizing structure $\tilde{H}_l = \sigma_l \otimes e^{is_lH}Ae^{-is_lH}$ with
$\sigma_l = w_l\left(\sigma_x \operatorname{Re} f\left(s_l\right)+\sigma_y \operatorname{Im} f\left(s_l\right)\right)$, allowing us to express the evolution as
\begin{align*} e^{-i\sqrt{\tau} \tilde{H}_l} = (I \otimes e^{is_lH}U_2)e^{-i \sqrt{\tau} (\sigma_l \otimes D)}(I \otimes U_2^\dag e^{-is_lH}),
\end{align*}
which is a product of three unitaries. The left and right terms can be implemented using random Clifford circuits and the Hamiltonian simulation of $H$. For the middle term, we apply Trotterization once more, reducing the problem to simulating Hamiltonians formed by Pauli strings. This can be achieved by rotating each $\sigma_x$ or $\sigma_y$ operator to $\sigma_z$, entangling the $n+1$ qubits by CNOT gates, and applying a phase gate (see \cite[Section 4.2]{Whitfield_2011} for details).

\subsubsection{The Relation between Mixing Time and Hamiltonian Spectrum}

Next, we analyze the mixing time of the Clifford-random Davies generator. The randomness of the Hermitian operator $A$ enables us to characterize the convergence of the dynamics $\rho_t$. 

\begin{theorem} \label{thm-Davies}
Under the assumptions of Theorem \ref{thm-construction-Davies}, the solution $\rho_t$ to \eqref{lindblad} satisfies:
\begin{itemize}
\item Let $p_i(t) =\EE \bra{\psi_i} \rho_t \ket{\psi_i}$ be its diagonal part. The evolution of $p_i(t)$ satisfies the Pauli Master equation:
\begin{align}\label{pauli}
\frac{\d p_i}{ \d t} = \xi^2 \sum_{k \ne i}  \gamma(\lambda_i - \lambda_k) p_k - \gamma(\lambda_k - \lambda_i) p_i \,.
\end{align}
\item Let $Q_{ij}(t)= \EE \bra{\psi_i} \rho_t \ket{\psi_j}$ with $i \neq j$ be the off-diagonal part. We have
\begin{align} \label{coherencef}
\frac{\d Q_{ij}}{ \d t} = -\frac{\xi^2}{2} \sum_k \left(\gamma(\lambda_k - \lambda_i) + \gamma(\lambda_k - \lambda_j) \right) Q_{ij}\,.
\end{align}
\end{itemize}
\end{theorem}

The proof is deferred to Appendix \ref{proof-fast-mixing}. Equation \eqref{pauli} shows that the diagonal part of the dynamics is governed by a Pauli master equation, which corresponds to a classical Metropolis-Hastings algorithm on the spectrum domain. Meanwhile, \eqref{coherencef} indicates that each off-diagonal entry $\ket{\psi_i}\bra{\psi_j}$ is an eigenstate of the dynamics, with its coefficient decaying exponentially at a rate determined by the spectral structure of $H$ and the function $\gamma$. 

Since the equations \eqref{pauli}-\eqref{coherencef} for a given $\gamma$ depend only on the eigenvalues of the Hamiltonian, the system's spectrum fully determines the mixing time of the random Davies generator.  In particular, we show that {when $\gamma$ is defined using the Metropolis transition rate}, the spectral gap of $\bar{\mc{L}}$ is lower-bounded by the density of low-energy states of $H$. 
 
\begin{theorem} \label{fast-mixing}
Let $r_\beta(H)$ denote the ratio of the states within an energy window above the ground state:
\begin{align} \label{low-energy}
r_\beta(H) = \frac{\#\left\{\ket{\psi_i}: \lambda_i \leq \lambda_{\min }\left({H}\right)+ 1/\beta \right\}}{N}\,, \q \text{{for any $\beta > 0$}}\,,
\end{align}
Choosing $\gamma(\omega) =  \min(1,e^{-\beta \omega})$,  the spectral gap of the dynamic \eqref{lindblad} satisfies
\begin{align*}
\eta = \Omega(r_\beta(H))\,.
\end{align*}
\end{theorem}
The proof strategy involves analyzing the convergence of the diagonal and off-diagonal terms separately. The convergence of the diagonal part is studied using conductance-based techniques for classical Markov chains \cite{10.1214/aoap/1177005980}. Meanwhile, the off-diagonal terms exhibit exponential decay on an entrywise basis. We show that the convergence rates of both parts are determined by the density of low-energy states defined in \eqref{low-energy}.

When the low-energy density of the Hamiltonian is non-negligible, the lower bound provided in Theorem \ref{fast-mixing} indicates rapid mixing.  A notable example of this scenario is a strongly interacting system that exhibits ergodic behavior. Under certain scenarios, such systems are believed to satisfy the Eigenstate Thermalization Hypothesis (ETH), and their spectra are well-described by the Gaussian Unitary Ensemble (GUE) random matrix theory \cite{D_Alessio_2016}.  Assuming that the spectral density of the system approximately follows the semicircle law (with appropriate normalization), we have
$$ \frac{\#\left\{\ket{\psi_i}: \lambda_i \in [x,x+dx]\right\}}{N} \approx \frac{\sqrt{4-x^2}}{2\pi}. $$
Thus, the low-energy density $r_\beta(H)$ in \eqref{low-energy} scales as:
\begin{align} \label{low-energy-bound}
\begin{aligned}
     r_\beta(H) & = \frac{\#\left\{\ket{\psi_i}: \lambda_i \leq \lambda_{\min }(H)+ 1/\beta \right\}}{N}  \\ & \approx \int_{-2}^{-2+1/\beta} \frac{\sqrt{4-x^2}}{2\pi} \d x  =  \Omega(\min(1, \beta^{-3/2}))\,,  \q \text{{for any $\beta > 0$}}\,.
\end{aligned}
\end{align}
In this case, combining \eqref{low-energy-bound} with Theorem \ref{fast-mixing}, the spectral gap of the dynamics scales polynomially with the inverse temperature $\beta$. This ensures polynomial-time convergence, as long as {the inverse temperature $\beta$ is not exponentially large in the system size $n$ (i.e., the system is not at very low temperature).}

\subsubsection{Provable Efficient Quantum Gibbs Samplers}
So far, we have discussed both the implementation and the mixing time of the dynamics driven by the Clifford-random Davies generator. From the complexity theory perspective, a natural question is: \emph{Does our algorithm provide a provably efficient quantum Gibbs sampling method for a class of Hamiltonians?} To achieve this, the Hamiltonian must satisfy the following conditions:
\begin{itemize}
\item The Hamiltonian $H$ is efficiently simulatable. This typically requires the Hamiltonian to be local or sparse.
\item The low-energy density of $H$ is non-negligible, i.e.,  $r_\beta(H)$ is not exponentially small. This property is often observed in strongly interacting systems exhibiting chaotic and ergodic behavior.
\end{itemize}

An example of a Hamiltonian that satisfies the sparsity condition and exhibits a semicircle spectrum is the random Pauli string model studied in \cite{chifang2023sparse}:
\begin{align}\label{random-pauli}
\begin{aligned}
    & {H}_{\mathrm{PS}}:=\sum_{j=1}^m \frac{r_j}{\sqrt{m}} {\sigma}_j \\ & \text{with } {\sigma}_j \stackrel{\mathrm{iid}}{\sim}\left\{{I}, {\sigma}_x, {\sigma}_y, {\sigma}_z\right\}^{\otimes n},\  r_j \stackrel{\mathrm{iid}}{\sim} \mathrm{Unif}\{+1,-1\}\,.
\end{aligned}
\end{align}
\cite[Theorem II.1]{chifang2023sparse} shows that for sufficiently large $m$, the low-energy density of the random Pauli string model is bounded below by $\Omega(\beta^{-3/2})$, consistent with the low-energy density of {the GUE in the large dimension limit described by the} semicircle law. 

\begin{lemma}[{Restatement of \cite[Theorem II.1]{chifang2023sparse}}]\label{lemma-random-Hamiltonian}
There exist absolute constants \( C_1, C_2, C_3 \) such that for any \( 1 \leq \beta \leq 2^{n / C_1} \), if \( H_{PS} \) is drawn from the Pauli string ensemble \eqref{random-pauli} with  
\( m \geq C_2 n^5 \beta^4 \), then with probability at least \( 1 - e^{-C_3 n^{1/3}} \), we have  
\[
r_\beta(H) = \Omega(\beta^{-3/2})\,.
\]
 \end{lemma}
Therefore, we conclude that for the random Pauli string model with an appropriately chosen polynomially large $m$, implementing the Clifford-random Davies generator using the randomized algorithm results in an efficiently convergent quantum Gibbs sampler. {We also refer interested readers to the more recent work \cite{ramkumar2024mixing} for additional examples of random Hamiltonians with non-negligible low-energy density and to \cite{brunner2412lindblad} for a generalized mixing time analysis under the eigenstate thermalization hypothesis.}

\section{Conclusion and Future Works}
In this paper, we have studied a qDRIFT-type method for simulating Lindblad dynamics and quantum Gibbs sampling. The key advantage of this randomized algorithm is its ability to simplify the simulation of the evolution of an ensemble of Lindbladians into the simulation of a single Lindbladian. For quantum Gibbs sampling, we introduced the Clifford-random Davies generator, which is both efficiently implementable and comes with an explicit spectral gap lower bound. We next discuss several future directions that could further improve the effectiveness and broaden the applicability of our methods.

\paragraph*{Mixing Time of Few-Body Sparse Random Hamiltonians.}  
Our work shows that the global randomness in the Clifford-random Davies generator enables rapid convergence for certain strongly interacting Hamiltonian systems with $n$-body interactions. However, many physically relevant strongly interacting models typically involve few-body interactions, for which the GUE approximation may break down. {For example, in the Sachdev-Ye-Kitaev (SYK) model, the global spectral distribution approaches a Gaussian in the limit of a large number of fermions ($n \to \infty$). For finite $n$, however, the spectrum is described by $Q$-Hermite polynomials,} with the tails of the average spectral density well-approximated by a semicircle law \cite{garcia2016spectral,garcia2017analytical,feng2019spectrum}. This discrepancy may cause our algorithm to underperform. A promising direction for future research is to develop new randomized jump operators that achieve fast mixing for Hamiltonians whose spectra adhere to the semicircle law only locally.

\paragraph*{Higher-order randomized methods for Lindbladian simulation.}
As noted in Remark \ref{rem:compare}, while our randomized algorithm is simple and practical, its complexity scaling in time $T$ and accuracy $\epsilon$ is not state-of-the-art. Specifically, the $\mathcal{O}(T^2/\epsilon)$ scaling for the average channel is comparable to the first-order Trotter method but is surpassed by more advanced methods for Lindbladian simulation \cite{childs2016efficient, cleve2016efficient, li2023simulating}. 
{A promising direction for future work is to adapt higher-order randomized compilation methods to the Lindbladian setting. This includes approaches like qFLO \cite{watson2024randomly} and qSWIFT \cite{nakaji2023qswift}, which were recently developed for Hamiltonian simulation. The qFLO method, for instance, enhances the standard qDRIFT protocol by incorporating Richardson extrapolation during classical post-processing. Similarly, qSWIFT achieves higher-order scaling by explicitly computing the higher-order terms of the qDRIFT channel as a function of the time step, a procedure that also relies on classical post-processing.}

\paragraph*{The concentration-type analysis for the random channel.}
Compared to the convergence analysis of the random product of unitaries \cite{huang2020matrix}, a limitation of our analysis for the random channel is that we only establish an error bound for the second moment and lack a concentration result with an exponentially small probability tail, {which requires higher moment estimates}. This limitation stems from technical challenges: the presence of multiplicative noise in our error analysis prevents us from deriving a  $q$-th moment bound for large $q$. Addressing this issue and developing a concentration-type analysis for this method is an important direction for future work.
 \begin{acknowledgments}
Bowen Li is partially supported by National Key R$\&$D Program of China, Grant No. 2024YFA1016000. Jianfeng Lu is supported in part by the National Science Foundation via grant DMS-2309378. Lexing Ying is supported by the National Science Foundation via grant DMS-2208163 and the U.S. Department of Energy, Office of Science, Accelerated Research in Quantum Computing Centers, Quantum Utility through Advanced Computational Quantum Algorithms, grant no. DE-SC002557. 
 \end{acknowledgments}
 \appendix
\section{Basic Facts and Lemmas} \label{sec:app1}

This section collects some basic facts and auxiliary proofs. 

\begin{lemma}\label{Taylor}
Let $\mc{B}(X)$ be the space of bounded linear operators on a Banach space $X$ with norm denoted by $\| \cdot \|_{\mc{B}(X)}$. Then, for small $A \in \mc{B}(X)$ and any $k > 0$, there is a constant $C_k > 0$ such that 
\begin{align*}
     \Big\|e^A - \sum_{m=0}^k \frac{A^m}{m!} \Big\|_{\mc{B}(X)} \leq C_k\|A\|_{\mc{B}(X)}^{k+1}\,.
\end{align*}
\end{lemma}

\begin{corollary}\label{op-bound}
Let $\bar{\cL} := \EE_{a \sim \mu} [\cL_a]$ be the expected Lindbladian in \eqref{dynamic}, and denote by $\bar{\cP}_{\tau} := \EE_{a \sim \mu} [e^{\tau \cL_a}]$ the expected single-step Lindblad evolution. For any norm $\|\cdot\|_X$ on the space of bounded linear operators on $\cB(\cH)$, assuming that $\sup_{a \in \cA} \| \cL_a\|_X \le C $, we have 
\begin{align*}
    \big\|e^{\tau \bar{\cL}} - \bar{\cP}_\tau \big\|_{X} \lesssim C^2 \tau^2\,, 
\end{align*}
and for any fixed $a \in \mc{A}$, 
\begin{align*}
     \big\|e^{\tau \bar{\cL_a}} - \bar{\cP}_\tau \big\|_{X} \lesssim C \tau \,. 
\end{align*}
\end{corollary}

We next give the proof of Lemma \ref{op-bound1}. 

\begin{proof}
First, we show that for any $A,B \in \cB(\cH)$, we have 
\begin{align} \label{000}
\|\Gamma_{\sigma}^{-1} (AB)\|_{\sigma,1/2} \le \|\Gamma_{\sigma}^{-1} (A)\|_{\sigma,1/2} \|\sigma^{1/4} B \sigma^{-1/4} \|\,,
\end{align}
and 
\begin{align} \label{001}
\|\Gamma_{\sigma}^{-1} (BA)\|_{\sigma,1/2} \le \|\Gamma_{\sigma}^{-1} (A)\|_{\sigma,1/2} \|\sigma^{-1/4} B \sigma^{1/4} \|\,.
\end{align}
Indeed, 
\begin{align*}
\|\Gamma_{\sigma}^{-1} (AB)\|_{\sigma,1/2} &  = \|\sigma^{-1/4}AB\sigma^{-1/4}\|_F = \|\sigma^{-1/4}A \sigma^{-1/4} \sigma^{1/4} B\sigma^{-1/4}\|_F  \\ & \le \|\sigma^{-1/4}A \sigma^{-1/4}\|_F \|\sigma^{1/4} B\sigma^{-1/4}\| = \|\Gamma_\sigma^{-1}(A) \|_{\sigma,1/2}\|\sigma^{1/4} B\sigma^{-1/4}\|\,,
\end{align*}
\begin{align*}
\|\Gamma_{\sigma}^{-1} (BA)\|_{\sigma,1/2} &  = \|\sigma^{-1/4}BA\sigma^{-1/4}\|_F =\|\sigma^{-1/4}B\sigma^{1/4} \sigma^{-1/4} A\sigma^{-1/4}\|_F \\ &  \le \|\sigma^{-1/4}A \sigma^{-1/4}\|_F \|\sigma^{-1/4}B\sigma^{1/4} \| = \|\Gamma_\sigma^{-1}(A)\|_{\sigma,1/2} \|\sigma^{-1/4}B\sigma^{1/4} \|\,.
\end{align*}
Thus, we complete the proofs of \eqref{000} and \eqref{001}. Using these, we obtain
\begin{align*}
\| \Gamma_{\sigma}^{-1} \left(i [H_a,\rho]\right) \|_{\sigma,1/2} & \le \|\Gamma_{\sigma}^{-1}(H_a\rho) \|_{\sigma,1/2} +\|\Gamma_{\sigma}^{-1}(\rho H_a) \|_{\sigma,1/2}  \\ & \le 2\|\sigma^{1/4} H_a \sigma^{-1/4}\| \|\Gamma_\sigma^{-1} (\rho) \|_{\sigma,1/2} \le 2\Lambda \|\Gamma_\sigma^{-1} (\rho) \|_{\sigma,1/2}\,,
\end{align*}
and
\begin{align*}
\| \Gamma_{\sigma}^{-1} \left(V_a\rho V_a^\dag \right) \|_{\sigma,1/2} \le \|\sigma^{-1/4}V_a \sigma^{1/4}  \|\|\sigma^{1/4} V_a^\dag  \sigma^{-1/4}  \| \|\Gamma_{\sigma}^{-1} (\rho)\|_{\sigma,1/2} \le \Lambda \|\Gamma_{\sigma}^{-1} (\rho)\|_{\sigma,1/2}\,,
\end{align*}
as well as
\begin{align*}
\|\Gamma_{\sigma}^{-1}(V_a^\dag V_a \rho  + \rho V_a^\dag V_a ) \|_{\sigma,1/2} & \le \|\Gamma_{\sigma}^{-1}(V_a^\dag V_a \rho  ) \|_{\sigma,1/2} + \|\Gamma_{\sigma}^{-1}( \rho V_a^\dag V_a ) \|_{\sigma,1/2} \\
& \le 2\|\sigma^{1/4} V_a^\dag V_a \sigma^{-1/4}  \|\|\Gamma_{\sigma}^{-1} \rho\|_{\sigma,1/2} = 2\Lambda \|\Gamma_{\sigma}^{-1}(\rho) \|_{\sigma,1/2}\,. 
\end{align*} 
The proof is complete by combining the above estimates. 
\end{proof}


\section{The Hamiltonian Simulation-based Algorithm} \label{app-B}

In this section, we discuss the implementation of a single Lindbladian operator $\cL_a$ using a Hamiltonian simulation-based algorithm \cite{ding2023singleancilla}. 
Recall from \eqref{eq:singlelind} that the Lindbladian $\cL_a$ is given by 
\begin{align*}
    \cL_a (\rho) = -i[H_a,\rho ] + V_a\rho V_a^{\dag} - \frac{1}{2}\{V_a^\dag V_a,\rho \}\,.
\end{align*}
For each time step, we employ the simplest first-order \emph{Trotter decomposition} to simulate $e^{\tau \mc{L}_a}$, which implements the Hamiltonian evolution and the dissipative one alternately. 

\begin{lemma} \label{firsttrotter}
Let $\cL_H (\rho) := -i[H,\rho]$ and $\cL_V(\rho) := V\rho V^\dag - \frac{1}{2} \{\rho, V^\dag V\}$ for given Hamiltonian $H$ and jump $V$, which satisfy $\cL_H (\si) = \cL_V (\si) = 0$.
\begin{itemize}
\item Under Assumption \ref{as0}, for any $\tau \ge 0$, we have 
\begin{align*}
    \big\| e^{\tau \cL_H }e^{\tau \cL_V} -  e^{\tau (\cL_H + \cL_V) } \big\|_{\diamond} = \Or(\lambda^2 \tau^2)\,.
\end{align*}
\item Under Assumption \ref{as1}, for any $\tau \ge 0$, we have  
\begin{align*}
\big\| e^{\tau \cL_H }e^{\tau \cL_V} -  e^{\tau (\cL_H + \cL_V) } \big\|_{(\sigma,-1/2) \to (\sigma,-1/2)} = \Or(\Lambda^2 \tau^2)\,.    
\end{align*}
\end{itemize}
\end{lemma}

\begin{proof}
This follows from the standard error analysis of the first-order Lie-Trotter formula, as presented in \cite[Proposition 9]{Childs_2021}.  {Using the variation-of-parameters formula,} we find 
\begin{align*}
    e^{t \cL_H }e^{t \cL_V} & = e^{t (\cL_H + \cL_V)} + \int_0^t \mathrm{d} s \,  e^{(t - s) (\cL_H + \cL_V)} \big[e^{s \mc{L}_H}, \mc{L}_V\big] e^{s \mc{L}_V} \\
    & = e^{t (\cL_H + \cL_V)} + \int_0^t \mathrm{d} s \,  e^{(t - s) (\cL_H + \cL_V)} e^{s \mc{L}_H}\big(\mc{L}_V -  e^{- s \mc{L}_H} \mc{L}_V e^{s \mc{L}_H} \big) e^{s \mc{L}_V}\,.
\end{align*}
Note 
\begin{align*}
    \mc{L}_V -  e^{- s \mc{L}_H} \mc{L}_V e^{s \mc{L}_H} = \int_0^{s} \mathrm{d}\tau \, e^{- \tau \mc{L}_H}\big[\mc{L}_H, \mc{L}_V \big]e^{\tau \mc{L}_H}\,.
\end{align*}
It follows that 
\begin{align*}
     e^{t \cL_H }e^{t \cL_V}
    & = e^{t (\cL_H + \cL_V)} + \int_0^t \mathrm{d} s \, \int_0^{s} \mathrm{d}\tau \, e^{(t - s) (\cL_H + \cL_V)} e^{s \mc{L}_H}  e^{- \tau \mc{L}_H}\big[\mc{L}_H, \mc{L}_V \big]e^{\tau \mc{L}_H} e^{s \mc{L}_V}\,.
\end{align*}
By the contractivity of $e^{t \cL_H}$ for $t \in \mathbb{R}$ and $e^{t \cL_V}$ for $t \ge 0$, we derive 
\begin{align*}
    \big\| e^{t \cL_H }e^{t \cL_V} - e^{t (\cL_H + \cL_V)} \big\|_\diamond & \le \int_0^t \mathrm{d} s \, \int_0^{s} \mathrm{d}\tau \, e^{(t - s) (\cL_H + \cL_V)} e^{s \mc{L}_H}  e^{- \tau \mc{L}_H}\big[\mc{L}_H, \mc{L}_V \big]e^{\tau \mc{L}_H} e^{s \mc{L}_V} \\
    & \le \frac{t^2}{2} \big\| \big[\mc{L}_H, \mc{L}_V \big] \big\|_\diamond \lesssim \lad^2 t^2\,.
\end{align*}
Similarly, we have $\| e^{t \cL_H }e^{t \cL_V} - e^{t (\cL_H + \cL_V)}\|_{(\sigma,-1/2) \to (\sigma,-1/2)} \lesssim \Lambda^2 t^2$. 
\end{proof}

 {From Lemma \ref{firsttrotter}, the single-step evolution $e^{t \mc{L}_a}$ can be approximately simulated by $e^{\tau \cL_{H_a}}e^{\tau \cL_{V_a}}$. We assume efficient quantum access to the Hamiltonian evolution $e^{\tau \cL_{H_a}}$.}
Next, we present the Hamiltonian simulation-based algorithm to implement the dissipative dynamics $e^{\tau \cL_V}$ with a single jump $V$.  Define the dilated operator
 \begin{equation} \label{dilatedop}
   \widetilde{V} := \begin{pmatrix}
   0 & V^\dagger \\
   V & 0
   \end{pmatrix},
   \end{equation}
which acts on the joint system-ancilla Hilbert space.
We prepare an ancilla qubit in state $|0\rangle$ and perform Hamiltonian simulation of $\widetilde{V}$ on the extended space for time $\sqrt{\tau}$, and then trace out the ancilla qubit:
\[
\Qcircuit @C=1em @R=1.5em {
    & \lstick{|0\rangle} & \multigate{1}{\exp(-i\widetilde{V}\sqrt{\tau})} & \meter & \rstick{\text{Discard}} \qw \\
    & \lstick{|\psi\rangle} & \ghost{\exp(-i\widetilde{V}\sqrt{\tau})} & \qw & \qw
}
\]
The following result verifies that this implementation satisfies the local truncation error conditions (Assumptions \ref{as-alg0} and \ref{as-alg1}) used in the analysis. 

\begin{lemma}
Let $\rho_t$ be the Lindblad dynamics defined by a single jump $V$: 
\begin{equation*}
    \frac{d\rho}{dt} = \cL_V(\rho) := V\rho V^\dag - \frac{1}{2}\{V^\dag V, \rho\},
\end{equation*}
and define the quantum channel $\cF_t$ for $t > 0$, 
\begin{equation*}
   \cF_t(\rho) := \mathrm{Tr}_a\left[e^{\sqrt{t}{\cL}_{\widetilde{V}}}\left(|0\rangle\langle 0| \otimes \rho \right)  \right] = \operatorname{Tr}_a \left[ e^{-i \widetilde{V} \sqrt{t}}(|0\rangle\langle 0| \otimes \rho) e^{i \widetilde{V} \sqrt{t}} \right]\,,
\end{equation*}
where $\widetilde{V}$ is given in \eqref{dilatedop} and ${\cL}_{\widetilde{V}}(\rho):= -i[\widetilde{V}, \rho]$, and $\mathrm{Tr}_a$ denotes the partial trace over ancilla qubit. 
\begin{enumerate}
    \item (Diamond norm bound) Under Assumption~\ref{as0},  {i.e., $\|V\|^2 \le \lad$, for $\tau = \Or(\lambda^{-1})$,}
    \begin{equation}  \label{diamond-bound}
        \|\cF_\tau - e^{\tau \cL_V}\|_\diamond = \Or(\lambda^2 \tau^2).
    \end{equation}  
    \item (Weighted $\ell^2$-norm bound) Under Assumption~\ref{as1},  {i.e., $\| \Delta_{\sigma}^{1/4} V \|^2 \le \Lambda$ and $\| \Delta_{\sigma}^{1/4} V^\dag  \|^2 \le \Lambda$, for $\tau = \Or(\Lambda^{-1})$,}
    \begin{equation} \label{chisquare-bound}
        \|\cF_\tau(\rho) - e^{\tau \cL_V}(\rho)\|_{\sigma, - 1/2} = \Or\left(\Lambda^2 \tau^2\right) \|\rho\|_{\sigma,- 1/2}\,.
    \end{equation}
     {Here $\sigma$ is an arbitrary full-rank quantum state.}  
\end{enumerate}
\end{lemma}

\begin{proof}
We begin with the third-order Taylor expansion of the simulated evolution:
\begin{equation}
e^{\sqrt{\tau} \cL_{\widetilde{V}}} + \Or(\tau^2) = T_{3,\cL_{\widetilde{V}}} := I + \sqrt{\tau} \cL_{\widetilde{V}} + \frac{1}{2}\tau \cL_{\widetilde{V}}^2 + \frac{1}{6}\tau^{\frac{3}{2}} \cL_{\widetilde{V}}^3\,,
\end{equation}
where $\cL_{\widetilde{V}}(\cdot) = -i[\widetilde{V}, \cdot]$. It is easy to see 
\begin{align*}
\mathrm{Tr}_a\left(\cL_{\widetilde{V}}(|0\rangle\langle 0| \otimes \rho)\right) & = 0\,,  \\
\mathrm{Tr}_a\left(\cL_{\widetilde{V}}^2(|0\rangle\langle 0| \otimes \rho)\right) & = \cL_V (\rho)\,,  \\
\mathrm{Tr}_a\left(\cL_{\widetilde{V}}^3(|0\rangle\langle 0| \otimes \rho)\right) & = 0\,.
\end{align*}
This allows us to define 
\begin{equation*}
\hat{\cF}_\tau (\rho) := \mathrm{Tr}_a\left(T_{3,\cL_{\widetilde{V}}}(|0\rangle\langle 0| \otimes \rho)\right) = (I + \tau \cL_V)\rho.
\end{equation*}
The simulation error decomposes as:
\begin{equation*}
\cF_\tau - e^{\tau \cL_V} = (\cF_\tau - \hat{\cF}_\tau) + (e^{\tau \cL_V} - (I + \tau \cL_V))\,.
\end{equation*}

For the diamond norm bound, we apply the triangle inequality to find
\begin{align}
\|\cF_{\tau} - e^{\tau \cL_V}\|_{\diamond} \leq \|\cF_{\tau} - \hat{\cF}_{\tau}\|_{\diamond} + \|e^{\tau \cL_V} - (I + \tau \cL_V)\|_{\diamond}.
\end{align}
Lemma~\ref{Taylor} provides the bound $\|e^{\tau \cL_V} - (I + \tau \cL_V)\|_{\diamond} = \Or(\lambda^2 \tau^2)$, so we focus on analyzing $\|\cF_\tau - \hat{\cF}_\tau\|_{\diamond}$.  
Using the definition of the diamond norm and properties of partial traces, we establish:
\begin{align}
\|\cF_\tau - \hat{\cF}_\tau\|_{\diamond} &= \sup_{\substack{\rho \in \cB(\cH^{\otimes 2}) \\ \|\rho\|_1 = 1}}  \|((\cF_\tau - \hat{\cF}_\tau) \otimes {\rm id})\rho\|_1 \nonumber \\
&\leq \sup_{\substack{\rho \in \cB(\mathbb{C}^2 \otimes \cH^{\otimes 2}) \\ \|\rho\|_1 = 1}} \|\mathrm{Tr}_a[((T_{3,\widetilde{V}} - e^{\sqrt{\tau}\cL_{\widetilde{V}}}) \otimes {\rm id})\rho]\|_1 \nonumber \\
&\leq \|T_{3,\widetilde{V}} - e^{\sqrt{\tau}\cL_{\widetilde{V}}}\|_{\diamond} = \Or(\lambda^2 \tau^2)\,,
\end{align}
where the second inequality follows from the contractivity of partial trace under trace norm, and the final bound comes from Lemma~\ref{Taylor} again with the estimate
\begin{align*}
    \|\cL_{\widetilde{V}}\|^4 = \Or(\|\widetilde{V}\|^4) = \Or(\|V\|^4) = \Or(\lambda^2)\,.    
\end{align*}
This completes the proof of \eqref{diamond-bound}.

We now establish the weighted $\ell^2$-norm error bound. Applying the triangle inequality and Lemma~\ref{Taylor} with Lemma \ref{op-bound1}, we decompose the error as:
\begin{align}
\|(\cF_{\tau} - e^{\tau \cL_V})\rho\|_{\sigma,-1/2} &\leq \|(\cF_{\tau} - \hat{\cF}_{\tau})\rho\|_{\sigma,-1/2} + \|(e^{\tau \cL_V} - (I + \tau \cL_V))\rho\|_{\sigma,-1/2} \nonumber \\
&\leq \|(\cF_{\tau} - \hat{\cF}_{\tau})\rho\|_{\sigma,-1/2} + \Or(\Lambda^2 \tau^2) \|\rho\|_{\sigma,-1/2}.
\end{align}
To bound $\|(\cF_{\tau} - \hat{\cF}_{\tau})\rho\|_{\sigma,-1/2}$, we analyze its Frobenius norm representation:
\begin{align}
& \left\| \mathrm{Tr}_a\left[\left( e^{-\sqrt{\tau} \cL_{\widetilde{V}}} - T_{3,\cL_{\widetilde{V}}}\right) (\ket{0}\bra{0} \otimes \rho) \right] \right\|_{\sigma,-1/2} \notag \\
 = & \left\|\sigma^{-1/4}  \mathrm{Tr}_a \left[\left(e^{-\sqrt{\tau} \cL_{\widetilde{V}}} - T_{3,\cL_{\widetilde{V}}}\right) (\ket{0}\bra{0} \otimes \rho)\right]\sigma^{-1/4} \right\|_{F} \notag  \\
 \lesssim & \left\|I \otimes \sigma^{-1/4}  \left(e^{-\sqrt{\tau} \cL_{\widetilde{V}}} - T_{3,\cL_{\widetilde{V}}}\right) (\ket{0}\bra{0} \otimes \rho) I \otimes \sigma^{-1/4} \right\|_{F}\,.  \label{100}
 \end{align} 
The key technical claim is that \eqref{100} is bounded by $\Or(\Lambda^2\tau^2)\|\Gamma_\sigma^{-1}\rho\|_{\sigma,1/2}$. For this, by Lemma~\ref{Taylor}, it suffices to verify the following operator norm bound:
\begin{equation}
\big\|(I \otimes \sigma^{-1/4})\cL_{\widetilde{V}}(\tilde{\rho})(I \otimes \sigma^{-1/4})\big\|_F \lesssim \sqrt{\Lambda}\big\|(I \otimes \sigma^{-1/4})\tilde{\rho}(I \otimes \sigma^{-1/4})\big\|_F
\end{equation}
for any $\tilde{\rho} \in \cB(\mathbb{C}^2 \otimes \cH)$. This follows from:
\begin{align*}
&\left\|(I \otimes \sigma^{-1/4})\cL_{\widetilde{V}}(\tilde{\rho})(I \otimes \sigma^{-1/4})\right\|_F \\
\leq  & \left\|(I \otimes \sigma^{-1/4})\widetilde{V}\tilde{\rho}(I \otimes \sigma^{-1/4})\right\|_F + \left\|(I \otimes \sigma^{-1/4})\tilde{\rho}\widetilde{V}(I \otimes \sigma^{-1/4})\right\|_F \\
\leq  & \left(\left\|(I \otimes \sigma^{-1/4})\widetilde{V}(I \otimes \sigma^{1/4})\right\| + \left\|(I \otimes \sigma^{1/4})\widetilde{V}(I \otimes \sigma^{-1/4})\right\|\right) \left\|(I \otimes \sigma^{-1/4})\tilde{\rho}(I \otimes \sigma^{-1/4})\right\|_F \\
\lesssim & \sqrt{\Lambda}\left\|(I \otimes \sigma^{-1/4})\tilde{\rho}(I \otimes \sigma^{-1/4})\right\|_F\,,
\end{align*}
where we used the submultiplicativity of the Frobenius norm and Assumption~\ref{as1}. This completes the proof of the weighted $\ell^2$-norm error bound \eqref{chisquare-bound}. 
\end{proof}

\section{Proof of Theorem \ref{thm:aver_err}} \label{sec:prof1}
\begin{proof}
 Let $\bar{\cP}_\tau = \EE_{a \sim \mu}[e^{\tau \cL_a}]$ be the average Lindblad evolution channel. 
 First, we apply Corollary \ref{op-bound} and Lemma \ref{op-bound0} to bound the distance between $\bar{\cP}_\tau$ and $e^{\tau \bar{\cL}}$:
 $$
 \big\| \bar{\cP}_\tau - e^{\tau \bar{\mathcal{L}}}\big\|_{\diamond} = \Or(\lambda^2 \tau^2)\,.$$
Combining this with the condition \eqref{local-truncation} and using the triangle inequality, the local truncation error of the average channel is bounded by
\begin{align} \label{local-error}
\begin{aligned}
\big\|e^{\tau \bar{\cL}} - \EE_{a \sim \mu} [\cF_{\tau}(\cL_a)]\big\|_\diamond &\le \big\|e^{\tau \bar{\cL}} -\bar{\cP}_\tau\big\|_\diamond  + \big\|\bar{\cP}_\tau - \EE_{a \sim \mu} [\cF_{\tau}(\cL_a)]\big\|_\diamond        \\
& \le \big\|e^{\tau \bar{\cL}} -\bar{\cP}_\tau\big\|_\diamond + \EE_{a \sim \mu} \big\|\cF_{\tau}(\cL_a) - e^{\tau \cL_a}\big\|_\diamond \\
& \lesssim \lambda^2 \tau^2\,.
\end{aligned}
\end{align}
Next, we split the $M$-step error as 
\begin{align}\label{123}
\begin{aligned}
 \big\|{\bar{\cE}}_{\tau,M}- e^{M\tau \bar{\cL}}\big\|_{\diamond}  \le  \big\|(e^{\tau\bar{\cL} } - \EE_{a \sim \mu} [\cF_{\tau}(\cL_a)]) {\bar{\cE}}_{\tau,M-1} \big\|_{\diamond}   + \big\| e^{\tau \bar{\cL}}(\bar{\cE}_{\tau,M-1}- e^{(M-1)\tau \bar{\cL}}  )\big\|_\diamond\,.
 \end{aligned}
\end{align}
The first term in the RHS of \eqref{123} is bounded by $\Or(\lambda^2 \tau^2)$ using \eqref{local-error}.
For the second term, we use the fact that $e^{\tau \bar{\cL}}  $ is contractive under the diamond norm: $ \| e^{\tau \bar{\cL}}(\bar{\cE}_{\tau,M} - e^{(M-1)\tau \bar{\cL}})  \|_\diamond  \le \|\bar{\cE}_{\tau,M-1}- e^{(M-1)\tau \bar{\cL}}   \|_\diamond $. Thus, there holds
\begin{align*}
\big\|\bar{\cE}_{\tau,M}- e^{M\tau \bar{\cL}} \big\|_\diamond = \Or(\lambda^2 \tau^2) + \big\|\bar{\cE}_{\tau,M-1}- e^{(M-1)\tau \bar{\cL}} \big\|_\diamond\,.
\end{align*}
By iteration, we derive 
$$
\big\|\bar{\cE}_{\tau,M} - e^{T\bar{\cL}} \big\|_{\diamond} = \Or(M \lambda^2 \tau^2) = \Or\left(\frac{\lambda^2 T^2}{M}\right)\,.
$$
We have completed the proof.
\end{proof}

\section{Proof of Theorem \ref{thm-simulation}}  \label{proof-thm-simulation}

\begin{proposition}\label{prop-local-chisquare}
Let $\rho_m:= \cE_{\tau,m}(\rho)$ be the output state of the randomized algorithm  {at the $m$-th step}, and $\rho(m\tau) := e^{m\tau \bar{\cL}}(\rho)$ be the exact state at time $m\tau$. Recall $e_{\tau,m}= \|\rho_m - \rho(m\tau) \|_{\sigma,-1/2}$, where $\sigma$ is a full-rank invariant state of $e^{t \bar{\cL}}$. Then, there exists constant $C > 0$,  {independent of $\rho$}, such that
\begin{align}\label{iteration-simulation} \EE\left[ e_{m+1}^2 | \cE_{\tau,m} \right] \le \left(1 +  C \Lambda^2 \tau^2 \right) e_{\tau,m}^2 + C\Lambda^2 \tau^2 \left(\chi^2(\rho(m\tau), \sigma) + 1\right).
\end{align}
\end{proposition}
\begin{proof}
Recalling $\bar{\cP}_{\tau} = \EE [e^{\tau \cL_a}]$ and 
using the identity:
\begin{equation*}
    \rho_{m+1} - \rho((m+1)\tau) = e^{\tau \bar{\cL}} (\rho_m - \rho(m\tau)) + (\cF_{\tau}(\cL_{m+1}) - e^{\tau \cL_{m+1}})\rho_m + (\bar{\cP}_{\tau} - e^{\tau \bar{\cL}})\rho_m + (e^{\tau \cL_{m+1}} - \bar{\cP}_{\tau})\rho_m\,, 
\end{equation*}
we decompose the error $e_{m+1} = \|\rho_{m+1} - \rho((m+1)\tau)\|_{\sigma,-1/2}$ as:
\begin{align*}
e_{m+1} \leq \underbrace{\big\|e^{\tau \bar{\cL}}(\rho_m - \rho(m\tau)) + (e^{\tau \cL_{m+1}} - \bar{\cP}_\tau) \rho_m\big\|_{\sigma, -1/2}}_{\rm I} + \underbrace{\big\|(\bar{\cP}_\tau - e^{\tau \bar{\cL}}) \rho_m + (\cF_\tau(\cL_{m+1}) - e^{\tau \cL_{m+1}}) \rho_m\big\|_{\sigma,-1/2}}_{\rm II}\,.
\end{align*}

For term ${\rm I}$, noting $\EE\left[e^{\tau \cL_a} - \bar{\cP}_\tau  \right]=0$, we have the orthogonal decomposition, by a direct computation, 
\begin{align} \label{eq:orth-decomp}
\mathbb{E}[{\rm I}^2 \mid \rho_m] =  \big\|e^{\tau \bar{\cL}}(\rho_m - \rho(m\tau))\big\|_{\sigma, - 1/2}^2  + \mathbb{E}\left[\big\|(e^{\tau \cL_{m+1}} - \bar{\cP}_\tau) \rho_m\big\|_{\sigma, - 1/2}^2 \mid \rho_m\right].
\end{align}
The first term in \eqref{eq:orth-decomp} is bounded by contractivity of $e^{\tau\bar{\cL}}$  {under $\|\cdot\|_{\sigma,- 1/2}$ in \eqref{weighcontra}}:
\begin{equation*}
\big\|e^{\tau \bar{\cL}}(\rho_m - \rho(m\tau))\big\|_{\sigma, - 1/2} \leq e_{\tau,m} := \|\rho_m - \rho(m\tau)\|_{\sigma,-1/2}\,,
\end{equation*}
since $\si$ is an invariant state of $e^{\tau \bar{\cL}}$. 
For the second term, Corollary~\ref{op-bound} yields
\begin{equation*}
\big\|(e^{\tau \cL_{m+1}} - \bar{\cP}_\tau) \rho_m\big\|_{\sigma,-1/2} = \Or(\Lambda \tau) \|\rho_m\|_{\sigma, - 1/2}\,.
\end{equation*}
Combining these two bounds gives 
\begin{equation*}
\mathbb{E}[{\rm I}^2 \mid \rho_m] \le e_{\tau,m}^2 + \Or(\Lambda^2 \tau^2) \| \rho_m\|_{\sigma, - 1/2}^2\,.
\end{equation*}

For term ${\rm II}$, by the triangle inequality, we have 
\begin{align} \label{eq:term_II_decomposition}
\begin{aligned}
&  \big\|(\bar{\cP}_\tau - e^{\tau \bar{\cL}}) \rho_m + (\cF_\tau(\cL_{m+1}) - e^{\tau \cL_{m+1}}) \rho_m \big\|_{\sigma, - 1/2}^2 \\
\lesssim\, & \big\| (\cF_\tau(\cL_{m+1}) - e^{\tau \cL_{m+1}}) \rho_m \big\|_{\sigma, - 1/2}^2 + \big\|(\bar{\cP}_\tau - e^{\tau \bar{\cL}}) \rho_m\big\|_{\sigma, - 1/2}^2 \\
\lesssim\, & \Lambda^4 \tau^4 \|\rho_m\|_{\sigma, - 1/2}^2,
\end{aligned}
\end{align}
where the last inequality follows from Assumption~\ref{as-alg1} and Corollary~\ref{op-bound}. Combining the bounds for terms ${\rm I}$ and ${\rm II}$, and applying the inequality $\|a+b\|^2 \le (1+\tau^2\Lambda^2)\|a\|^2 + \left(1+\frac{1}{\tau^2\Lambda^2}\right)\|b\|^2$, we obtain
\begin{align} \label{eq:error_recursion}
\begin{aligned}
\mathbb{E}\left[e_{m+1}^2 \mid \mathcal{E}_{\tau, m}\right] 
& \leq \left(1+\tau^2 \Lambda^2\right) \mathbb{E}\left[{\rm I}^2 \mid \rho_m\right] + \left(1+\frac{1}{\tau^2 \Lambda^2}\right) \mathbb{E}\left[{\rm II}^2 \mid \rho_m\right] \\
& \leq (1 + \Lambda^2 \tau^2)\left(e_{\tau,m}^2 + \Or(\Lambda^2 \tau^2) \|\rho_m \|_{\sigma, - 1/2}^2\right) + \Or(\Lambda^2 \tau^2)\|\rho_m \|_{\sigma, - 1/2}^2\,.
\end{aligned}   
\end{align}

Next, we bound the term $\|\rho_m \|_{\sigma, - 1/2}^2$:
\begin{align}\label{eq:norm_decomposition}
\begin{aligned}
\|\rho_m \|_{\sigma, - 1/2}^2 = \|\Gamma_{\sigma}^{-1}\rho_m \|_{\sigma,1/2}^2 
& \lesssim \left\|\Gamma_\sigma^{-1}\left(\rho_m - \rho(m\tau)\right)\right\|_{\sigma,1/2}^2 + \left\|\Gamma_\sigma^{-1}(\rho(m\tau) - \sigma)\right\|_{\sigma,1/2}^2 + \left\|\Gamma_\sigma^{-1} \sigma\right\|_{\sigma,1/2}^2 \\
& \lesssim e_{\tau,m}^2 + \chi^2(\rho(m\tau), \sigma) + 1\,.
\end{aligned}
\end{align}
Substituting \eqref{eq:norm_decomposition} into \eqref{eq:error_recursion}, we obtain the desired estimate:
\begin{equation*}
\mathbb{E}\left[ e_{m+1}^2 \mid \mathcal{E}_{\tau,m} \right] \leq \left(1 +  \Or(\Lambda^2 \tau^2)\right) e_{\tau,m}^2 + \Or(\Lambda^2 \tau^2)\left(\chi^2(\rho(m\tau), \sigma) + 1\right).
\end{equation*}
This completes the proof of the error recursion.
\end{proof}

\begin{proof}[Proof of Theorem~\ref{thm-simulation}]
Iterating the inequality \eqref{iteration-simulation} yields, for initial state $\rho$, 
\begin{align} \label{eq:final_bound}
\mathbb{E} [e_{\tau,M}(\rho)^2] \leq C \sum_{m=0}^{M-1} \left(1 + C\Lambda^2\tau^2\right)^{M - 1 - m} \Lambda^2 \tau^2\left(\chi^2(e^{m\tau \bar{\cL}}\rho,\sigma) + 1\right).
\end{align}
When $\tau \le \frac{1}{C \Lambda^2 T}$, we have
\begin{equation*}
1 + C \Lambda^2 \tau^2 = 1 + C \Lambda^2 \tau \frac{T}{M} \leq 1 + \frac{1}{M} \quad \text{and} \quad \left(1 + C\Lambda^2 \tau^2 \right)^{M- 1 -m} \leq \left(1+ \frac{1}{M}\right)^M \leq e.
\end{equation*}
Thus, the right-hand side of \eqref{eq:final_bound} is bounded by 
\begin{equation*}
\mathbb{E} [e_{\tau,M}(\rho)^2] \leq e M  (C \tau^2 \Lambda^2) \left(\chi^2(\rho,\sigma) + 1\right) \lesssim  \frac{\Lambda^2 T^2}{M}\left(1 + \chi^2(\rho,\sigma)\right),
\end{equation*}
where we have used the contractivity of $e^{m\tau \bar{\cL}}$ for $\chi^2$-divergence.
\end{proof}
\begin{proof}[Proof of \eqref{improvest}]
   {Under Assumption \ref{as2}, similarly to the above arguments, by Proposition  \ref{prop:gap}, we derive
\begin{equation} \label{prfimpest}
    \begin{aligned}
          \mathbb{E} [e_{\tau,M}(\rho)^2] & \lesssim  \frac{\Lambda^2 T^2}{M} + C \sum_{m=0}^{M-1} \left(1 + C\Lambda^2\tau^2\right)^{M - 1 - m} \Lambda^2 \tau^2 e^{- 2 \eta m \tau} \chi^2(\rho,\sigma) \\
    & \lesssim  \frac{\Lambda^2 T^2}{M} + \Lambda^2 \tau^2 \sum_{m=0}^{M-1}  e^{- 2 \eta m \tau} \chi^2(\rho,\sigma) \\
    & \lesssim  \frac{\Lambda^2 T^2}{M} +  \frac{\Lambda^2 T^2}{M^2} \frac{1 - e^{-2 \eta T}}{1 - e^{- 2 \eta \tau}}  \chi^2(\rho,\sigma)\,.        
    \end{aligned}
\end{equation}
This completes the proof. }
\end{proof}

\section{Proof of Theorem \ref{thermal-convergence} } \label{proof-thermal-convergence}
\begin{proposition}
Under Assumptions \ref{as1},\ref{as-alg1},\ref{as2}, for any state $\rho \in \cD(\cH)$, we have,  {for $\tau = \Or(\Lambda^{-1})$},
\begin{align} \label{local-average}
 \sqrt{\chi^2( \EE\cF_\tau(\cL_a)\rho,\sigma )}   \le \left(e^{-\eta \tau } + C \tau^2 \Lambda^2 \right)\sqrt{\chi^2({\rho},\sigma)} + C \tau^2 \Lambda^2\,,
\end{align}
for some constant $C > 0$.
\end{proposition}
\begin{proof}
Recalling $\bar{\cP}_{\tau} = \EE [e^{\tau \cL_a}]$ and using 
\begin{align*}
    \cF_\tau(\cL_a) = \cF_\tau(\cL_a) - \bar{\cP}_\tau + \bar{\cP}_\tau - e^{\tau \bar{\cL}} + e^{\tau \bar{\cL}}\,,    
\end{align*}
we split the error as
\begin{align} \label{00}
\begin{aligned}
& \sqrt{\chi^2\left(\EE[\cF_\tau(\cL_a)(\rho)], \sigma\right)} \\
= & \left\|P_I^{\perp} \Gamma_\sigma^{-1}\left( (\EE\cF_\tau(\cL_a) - \bar{\cP}_\tau) + (\bar{\cP}_\tau - e^{\tau \bar{\cL}} ) + e^{\tau \bar{\cL}}   \right) {\rho}\right\|_{\sigma, 1 / 2} \\
\leq & \sqrt{\chi^2(e^{\tau \bar{\cL}}\rho, \sigma) } + \big\|\big(\bar{\cP}_\tau-e^{\tau \bar{\mathcal{L}}}\big) \rho\big\|_{\sigma, -1 / 2}  +\big\|\big(\bar{\cP}_\tau- \EE \cF_{\tau}(\cL_a) \big) {\rho}\big\|_{\sigma, -1 / 2}\,. \\
\end{aligned}
\end{align}
One can bound the second and third terms in \eqref{00} by Corollary \ref{op-bound} and Assumption \ref{as-alg1}: 
\begin{align*}
\big\|\big(\bar{\cP}_\tau-e^{\tau \bar{\mathcal{L}}}\big) \rho\big\|_{\sigma, -1 / 2}  +\big\|\big(\bar{\cP}_\tau- \EE \cF_{\tau}(\cL_a) \big) {\rho}\big\|_{\sigma, -1 / 2}\lesssim \Lambda^2 \tau^2 \|\rho\|_{\sigma,-1/2}\,,
\end{align*}
where $\|\rho\|_{\sigma,-1/2}^2 = 1+ \chi^2(\rho,\sigma)$. 
To bound the first term of \eqref{00}, by the spectral gap condition in Assumption \ref{as2} and Proposition  \ref{prop:gap}, we have
$$ \sqrt{\chi^2(e^{\tau \bar{\cL}}\rho, \sigma) }  \le e^{-\tau \eta} \sqrt{\chi^2(\rho, \sigma) }\,.$$
Combining the above estimates, we complete the proof.
\end{proof}

\begin{proof}[Proof of  \eqref{convergence-average-sampler}]
     Iterating 
\eqref{local-average} implies 
\begin{align*}
\sqrt{\chi^2\left({\bar{\mathcal{E}}}_{\tau, M}(\rho), \sigma\right)} & \le \left(e^{-\eta \tau } + C \tau^2 \Lambda^2 \right)^M\sqrt{\chi^2({\rho},\sigma)} + 
\sum_{k = 0}^{M-1} \left(e^{-\eta \tau } + C \tau^2 \Lambda^2 \right)^{k}
C \tau^2 \Lambda^2 \\
& \le e^{-\eta \tau M} \left(1 + C e^{\eta \tau} \tau^2 \Lambda^2 \right)^M\sqrt{\chi^2({\rho},\sigma)} + 
\sum_{k = 0}^{M-1} e^{-\eta \tau k} \left(1 + C e^{\eta \tau} \tau^2 \Lambda^2 \right)^k
C \tau^2 \Lambda^2 
\end{align*}
 {When $\tau \le \frac{1}{\eta}$ and $\tau \le \frac{1}{C e T \Lambda^2}$, we have 
\begin{align*}
    1 + C e^{\eta \tau} \tau^2 \Lambda^2 \le 1 + \frac{1}{M}\,, \quad  \left(1 + C e^{\eta \tau} \tau^2 \Lambda^2 \right)^k \le \big(1 + \frac{1}{M}\big)^M \le e\,,\ \forall k \le M\,.
\end{align*}
It follows that 
\begin{equation} \label{auxeq1}
    \begin{aligned}
        \sqrt{\chi^2\left({\bar{\mathcal{E}}}_{\tau, M}(\rho), \sigma\right)} &\lesssim e^{-\eta \tau M} \sqrt{\chi^2({\rho},\sigma)} + 
\sum_{k = 0}^{M-1} e^{-\eta \tau k} 
 \tau^2 \Lambda^2 \\
 &\lesssim e^{-\eta \tau M} \sqrt{\chi^2({\rho},\sigma)} + 
\frac{1}{1 - e^{-\eta \tau}} 
 \tau^2 \Lambda^2 \\
  &\lesssim e^{-\eta \tau M} \sqrt{\chi^2({\rho},\sigma)} + 
\frac{\tau \Lambda^2}{\eta}\,.
    \end{aligned}
\end{equation}
This completes the proof.}
\end{proof}

\begin{proposition} \label{prop:local-random}
Under Assumptions \ref{as1},\ref{as-alg1},\ref{as2}, for any state $\rho \in \cD(\cH)$, we have,  {for $\tau = \Or(\Lambda^{-1})$},
\begin{align} \label{local-random}
 \EE \left[ \chi^2(\cF_\tau(\cL_a)\rho,\sigma )  \right] \le \left(e^{-2\eta \tau } +C \tau^2 \Lambda^2 \right)\chi^2({\rho},\sigma) + C \tau^2 \Lambda^2\,.
\end{align}
for some constant $C > 0$.
\end{proposition}
\begin{proof}
Similarly to \eqref{000}, using 
\begin{align*}
    \cF_\tau(\cL_a) = e^{\tau \bar{\cL}} +  (\bar{\cP}_\tau - e^{\tau \bar{\cL}} ) + (\cF_\tau(\cL_a) -e^{\tau \cL_a}) + (e^{\tau \cL_a} - \bar{\cP}_\tau)\,,    
\end{align*}
for a fixed state $\rho$, we have
\begin{align*} \chi^2(\cF_\tau(\cL_a)\rho,\sigma)  = \Big\|\big(e^{\tau \bar{\cL}} +  (\bar{\cP}_\tau - e^{\tau \bar{\cL}} ) + (\cF_\tau(\cL_a) -e^{\tau \cL_a}) + (e^{\tau \cL_a} - \bar{\cP}_\tau)\big)\rho - \sigma\Big\|_{\sigma, - 1/2}^2\,.
\end{align*}
Thanks to $\EE[e^{\tau \cL_a} - \bar{\cP}_\tau]=0$, the following orthogonal decomposition holds  
\begin{align} \label{1111}
& \EE \chi^2(\cF_\tau(\cL_a)\rho,\sigma) \notag  \\  = & \EE \underbrace{\left\|(e^{\tau \cL_a} - \bar{\cP}_\tau)\rho \right\|_{\sigma, - 1/2}^2}_{{\rm I}} + \EE \underbrace{\big\|  ( e^{\tau \bar{\cL}} +  (\bar{\cP}_\tau - e^{\tau \bar{\cL}} ) + (\cF_\tau(\cL_a) -e^{\tau \cL_a}))\rho - \si \big\|_{\sigma, - 1/2}^2}_{{\rm II}}\,.
\end{align}
For the term ${\rm I}$, by using the triangle inequality and Corollary \ref{op-bound}, we find 
\begin{align} \label{2222}
\begin{aligned}
  &  \left\|(e^{\tau \cL_a} - \bar{\cP}_\tau) \rho\right\|_{\sigma, - 1/2}^2  \le 2 \left\|(e^{\tau \cL_a} - \bar{\cP}_\tau) (\rho-\sigma) \right\|_{\sigma, - 1/2}^2 +2 \left\|(e^{\tau \cL_a} - \bar{\cP}_\tau) \sigma \right\|_{\sigma, - 1/2}^2  \\  
\lesssim &  \tau^2 \Lambda^2 \left(\left\|\rho-\sigma \right\|_{\sigma, - 1/2}^2 + 1 \right) \\
\lesssim &  \tau^2 \Lambda^2 \left( \chi^2(\rho,\sigma) + 1 \right). 
\end{aligned}
\end{align}
Next, we bound the term ${\rm II}$ in \eqref{1111}. We use $\|a+b\|^2 \le (1+\tau^2\Lambda^2)\|a\|^2 + \left(1+\frac{1}{\tau^2\Lambda^2}\right) \|b\|^2$ to split the error as 
\begin{align}\label{3333}
{ \rm II} \le (1+\tau^2\Lambda^2)\big\|e^{\tau \bar{\cL}}(\rho) - \si \big\|_{\sigma, - 1/2}^2 + \left(1+\frac{1}{\tau^2 \Lambda^2}\right) \big\|(\bar{\cP}_\tau - e^{\tau \bar{\cL}})\rho + (\cF_{\tau}(\cL_a) - e^{\tau \cL_a})\rho \big\|_{\sigma, - 1/2}^2,
\end{align}
Using the spectral gap condition in Assumption \ref{as2}, the first term in \eqref{3333} decays as 
\begin{align} \label{4444}
 {\big\|e^{\tau \bar{\cL}}(\rho) - \si \big\|_{\sigma, - 1/2}^2 = \chi^2(e^{\tau \bar{\cL}}(\rho), \sigma) \le e^{-2\eta \tau} \chi^2(\rho,\sigma)\,.}
\end{align}
Moreover, the second term in \eqref{3333} is of order $\Or(\tau^2)$ by combining Corollary \ref{op-bound} and Assumption \ref{as-alg1}: 
\begin{align} \label{5555}
\begin{aligned}
& \big\|(\bar{\cP}_\tau - e^{\tau \bar{\cL}})\rho + (\cF_{\tau}(\cL_a) - e^{\tau \cL_a})\rho  \big\|_{\sigma,- 1/2}^2 \\
\lesssim\,  & \big\|(\bar{\cP}_\tau - e^{\tau \bar{\cL}})(\rho - \si) + (\cF_{\tau}(\cL_a) - e^{\tau \cL_a})(\rho - \si)  \big\|_{\sigma,- 1/2}^2 + \big\|(\bar{\cP}_\tau - e^{\tau \bar{\cL}})\si + (\cF_{\tau}(\cL_a) - e^{\tau \cL_a})\si  \big\|_{\sigma,- 1/2}^2 \\
\lesssim \, &  {\Lambda^4 \tau^4 \left(\chi^2(\rho,\sigma) + 1 \right)}.
\end{aligned}
\end{align}
Combining \eqref{2222}--\eqref{5555} together, we obtain
\begin{align*}
\EE\left[ \chi^2(\cF_{\tau}(\cL_a)\rho,\sigma)\right] \le  \left(e^{-2\eta \tau} + C \Lambda^2 \tau^2\right)\chi^2(\rho,\sigma) + C \Lambda^2 \tau^2\,.
\end{align*}
We complete the proof of \eqref{local-random}.
\end{proof}

\begin{proof}[Proof of  \eqref{convergence-sampler}]
Similarly, iterating \eqref{local-random}, we find
\begin{align*}
    \EE\big[ {\chi^2(\cE_{\tau,M}(\rho),\sigma)} \big]  \le \left(e^{-2\eta \tau} + C \Lambda^2 \tau^2\right)^M\chi^2(\rho,\sigma) + \sum_{k = 0}^{M-1} \left(e^{-2\eta \tau} + C \Lambda^2 \tau^2\right)^k C \Lambda^2 \tau^2\,.
\end{align*}
 {When $\tau = \Or\big(\min\{\frac{1}{\eta}, \frac{1}{\Lambda}, \frac{1}{T \Lambda^2}\}\big)$, it follows that 
\begin{align*}
    \EE\big[ {\chi^2(\cE_{\tau,M}(\rho),\sigma)} \big]  \lesssim \e^{-2\eta \tau M} \chi^2(\rho,\sigma) + \frac{\tau \Lambda^2}{\eta}\,,
\end{align*}
by analogous calculations to those in \eqref{auxeq1}.}
\end{proof}

\section{Proof of Lemma \ref{lemma-random-design}}\label{app:random}
\begin{proof}
For fixed indices $i,j,k,l$, noting that $A_{ij}\overline{A_{lk}} = \mathrm{Tr}(U_2^{\otimes 2} D^{\otimes 2}U_2^{\dag \otimes 2} (\ket{\psi_j}\bra{\psi_i}\otimes \ket{\psi_l}\bra{\psi_k})  )$, we have
\begin{align*}
\EE \left[A_{ij} \overline{A_{lk}} \right] & = \EE \left[\mathrm{Tr}(U_2^{\otimes 2} D^{\otimes 2}U_2^{\dag \otimes 2} (\ket{\psi_j}\bra{\psi_i}\otimes \ket{\psi_l}\bra{\psi_k}))\right] \\ &  =  \EE_{U \sim \mathrm{Haar}} \left[\mathrm{Tr}(U^{\otimes 2} D^{\otimes 2}U^{\dag \otimes 2} (\ket{\psi_j}\bra{\psi_i}\otimes \ket{\psi_l}\bra{\psi_k}))\right] \\ & = \EE_{{U} \sim \mathrm{Haar}} \left[\bra{\psi_i} UDU^\dag  \ket{\psi_j}\bra{\psi_k}U DU^\dag  \ket{\psi_l}\right]\,.
\end{align*}
Let $(\ket{\phi_j})_{1\le j \le N}$ denote the computational basis. Then, $D$ is diagonal with random $\pm 1$ entries under this basis. Note that ${U}\ket{\psi_j}$ is identical to ${U}\ket{\phi_j}$ in distribution. Denote ${U}^{ij} =  \bra{\phi_i} {U}\ket{\phi_j} $ and  $r_j \sim \mathrm{Unif}\{+1,-1 \}$ with $\EE r_ir_j = \delta_{ij}$. It follows that 
\begin{align*}
& \EE_{{U} \sim \mathrm{Haar}} \left[\bra{\psi_i} U DU^\dag  \ket{\psi_j}\bra{\psi_k}U DU^\dag\ket{\psi_l}\right]
= \EE \left[\sum_{s=1}^N r_s {U}^{is}\overline{U^{js}} \sum_{s=1}^N r_s {U}^{ks}\overline{U^{ls}} \right] \\
 = & \EE \left[\sum_{s=1}^N  {{U}}^{is}\overline{U^{js}} {U}^{ks}\overline{U^{ls}}  \right] \\
 = & \begin{dcases}  0,  & \quad (i,j ) \ne ( l, k)\,, \\ \sum_{s=1}^N  \EE \left[|{U}^{is}|^2 |{U}^{js}|^2\right] =  \Theta(1/N), & \quad   (i,j ) =( l, k)\,.  
    \end{dcases}
\end{align*}
The proof is complete. 
\end{proof}

\section{Proof of Theorem \ref{thm-construction-Davies}} \label{proof-thm-construction-Davies}
\begin{proof}
We expand $K_a$ in \eqref{def:kintegral} with respect to the energy basis $\ket{\psi_i}$: 
\begin{equation}\label{expka}
    K_a = \sum_{i=1}^N \sum_{j=1}^N \sqrt{\gamma(\lambda_i - \lambda_j) } \ket{\psi_i}\bra{\psi_i}A_a\ket{\psi_j}\bra{\psi_j} = \sum_{\omega \in B_H} \sqrt{\gamma(\omega)} A_a(\omega)\,,    
\end{equation}
where $A_a(\omega)$ is defined in \eqref{def:aomega}. Thus, we have   
\begin{align} \label{explka}
    \cL_{K_a}(\rho) = \sum_{\omega,\omega' \in B_H} \sqrt{\gamma(\omega) \gamma(\omega')}\left( A_a(\omega) \rho A_a(\omega')^\dag - \frac{1}{2}\{A_a(\omega')^\dag A_a(\omega), \rho \}\right).
\end{align}
For any $1\le i,j\le N$, there holds, for any operator $A$, 
\begin{align*}
& A(\omega)\ket{\psi_i} \bra{\psi_j} A(\omega')^\dag = \sum_{\lambda_k-\lambda_i = \omega} \sum_{\lambda_{k'} - \lambda_j = \omega'} A_{ki}\overline{A_{k'j}}\ket{\psi_k}\bra{\psi_{k'}}\,, \\
& A(\omega')^\dag A(\omega)\ket{\psi_i}\bra{\psi_j} = \sum_{\lambda_{k'} - \lambda_k =  - \omega' } \sum_{\lambda_k - \lambda_i = \omega} \overline{A_{kk'}}A_{ki}\ket{\psi_{k'}}\bra{\psi_j}\,, \\
& \ket{\psi_i}\bra{\psi_j}A(\omega')^\dag A(\omega) = \sum_{\lambda_{k'} - \lambda_{k} = -\omega}\sum_{\lambda_k - \lambda_j = \omega'} \overline{A_{kj}}A_{kk'} \ket{\psi_i}\bra{\psi_{k'}}\,.
\end{align*}
When $\omega \ne \omega'$, by Assumption \ref{random-operator}, we have $\EE \left[A_{ki}\overline{A_{k'j}}\right] = \EE \left[\overline{A_{kk'}}A_{ki}\right] = \EE \left[\overline{A_{kj}}A_{kk'}\right] = 0$ for $A$ from the ensemble $\{A_a\}_{a \in \mc{A}}$, thus the expected summation over $\omega \ne \omega'$ in \eqref{explka}
vanishes. Therefore, we can conclude that 
 \begin{align*}
     \EE [\cL_{K_a}(\rho)] = \sum_{\omega \in B_H} \gamma(\omega) \EE_a\left( A_a(\omega) \rho A_a(\omega)^\dag - \frac{1}{2} \left\{A_a(\omega)^\dag A_a(\omega), \rho \right\}\right)
 \end{align*}
 is a Davies generator. 
\end{proof}

\section{Proof of Theorem \ref{thm-Davies}} \label{proof-thm-Davies}
\begin{proof}
We omit the subscript $a$ for $A_a$ and $K_a$ for simplicity. We first compute the expectation of $K \ket{\psi_i}\bra{\psi_j} K^\dag - \frac{1}{2}\{K^\dag K, \ket{\psi_i}\bra{\psi_j} \}$ for $i,j \in [N]$. Note that, by \eqref{expka},  
 $$
K\left|\psi_i\right\rangle=\sum_{k=1}^{N} \sqrt{\gamma(\lambda_k - \lambda_i) }A_{k i}\left|\psi_k\right\rangle, \quad K^{\dagger}\left|\psi_i\right\rangle=\sum_{k=1}^{N} \sqrt{\gamma(\lambda_i - \lambda_k)} \overline{A_{ik}}\left|\psi_k\right\rangle\,. $$
We have, according to Assumption \ref{random-operator}, 
\begin{align*}
\EE \left[K \ket{\psi_i}\bra{\psi_j}K^\dag\right] & = \sum_{k=1}^N \sum_{k'=1}^N  \sqrt{\gamma(\lambda_k - \lambda_i) \gamma(\lambda_{k'} - \lambda_j)} \EE \left[A_{k i}\overline{A_{k'j}}\right] \ket{\psi_k}\bra{\psi_{k'}} 
\\ & = \begin{cases} 
\xi^2 \sum_{k=1}^N \gamma(\lambda_k - \lambda_i) \ket{\psi_k}\bra{\psi_k},\quad & i=j\,,   \\
 0\,, \quad  & i \ne j\,.
\end{cases}
\end{align*}
and 
\begin{align*}
 \begin{aligned}
\EE \left[K^\dag K  \ket{\psi_i}\bra{\psi_j}\right] 
& = \sum_{k', k=1}^{N}\sqrt{\gamma(\lambda_k - \lambda_i
) \gamma(\lambda_k - \lambda_{k'}) } \EE \left[\overline{A_{k k'}} A_{k i} \right]\left|\psi_{k'}\right\rangle\left\langle\psi_j\right| \\
& = \xi^2 \left(\sum_{k=1}^{N} \gamma(\lambda_k - \lambda_{i})\right)\left|\psi_i\right\rangle\left\langle\psi_j\right|.
\end{aligned}
\end{align*}
Similarly, there holds
$$
\EE \left[\ket{\psi_i}\bra{\psi_j} K^\dag K \right] = \xi^2 \left(\sum_{k=1}^{N} \gamma(\lambda_k - \lambda_{j})\right)\left|\psi_i\right\rangle\left\langle\psi_j\right|\,.
$$

Combining the above computations, we conclude  
\begin{align*}
\EE \left[\cL_K \ket{\psi_i}\bra{\psi_i}\right] = \xi^2 \sum_{k \ne i} \gamma(\lambda_k - \lambda_i) \ket{\psi_k}\bra{\psi_k}  - \gamma(\lambda_k - \lambda_i) \ket{\psi_i}\bra{\psi_i}\,,
\end{align*}
and for $i \ne j$,
\begin{align*}
\EE \left[\cL_K \ket{\psi_i}\bra{\psi_j}\right] = -\frac{1}{2} \sum_k \xi^2\left(\gamma(\lambda_k - \lambda_i) + \gamma(\lambda_k - \lambda_j) \right) \ket{\psi_i}\bra{\psi_j}\,.
\end{align*}
We complete the proof of 
\eqref{pauli} and \eqref{coherencef}.

\end{proof}

\section{Proof of Theorem \ref{fast-mixing}} \label{proof-fast-mixing}

We begin by presenting the conductance-based technique for the mixing time analysis of classical random walks, used in the proof of Theorem \ref{fast-mixing}.
  
\begin{lemma} [{\cite[Proposition 1]{10.1214/aoap/1177005980}}] \label{spectral-bound}  
  Consider an irreducible, reversible Markov chain on a finite state space $X$, described by a transition matrix $P = (p_{ij})_{i,j\in X}$. Suppose that $p_{ij} > 0$ for all $i \neq j$, and denote by $\pi = (\pi_i)_{i\in X}$ the unique stationary distribution. Let $L=P-I$ be the generator for the associated continuous-time Markov semigroup. Then, the spectral gap of $L$ is bounded below by $\min_{i,j\in X,\, j\neq i} \frac{p_{ij}}{\pi_j}$. 
\end{lemma}

\begin{theorem} \label{conductance}
Consider a state space $X = \{1,\cdots,N \}$ equipped with the Gibbs distribution $\pi_i = \frac{1}{Z(\beta)}e^{-\beta \lambda_i}$, where $0 < \lambda_1 \le \lambda_2 \le \cdots \le \lambda_N$ are energy levels. For the Metropolis--Hastings transition probabilities:
    \[
    p_{ij} = \frac{1}{2N}\min\big\{1, e^{-\beta(\lambda_j-\lambda_i)}\big\}, \quad i\neq j, \quad 
    p_{ii} = 1 - \sum_{j\neq i} p_{ij},
    \]
    the spectral gap of $L = P - I$ is bounded below by
    \[
   \frac{1}{2N}\sum_{i=1}^N e^{-\beta(\lambda_i - \lambda_1)}\,.
    \]
\end{theorem}

\begin{proof}
We apply Lemma \ref{spectral-bound}, and find that the spectral gap $\lad(L)$ is bounded by $$\lad(L) \ge  \min_{i\neq j} \frac{p_{ij}}{\pi_{j}}\,.$$ Note that for any $i \neq j$, there holds 
\begin{align*}
     \frac{p_{ij}}{\pi_{j}} \ge \frac{\frac{1}{2N}e^{-\beta(\lambda_j - \lambda_1)}}{\pi_j}  = \frac{Z(\beta)}{2Ne^{-\beta\lambda_1}}\,.
\end{align*}
Thus, we complete the proof by $Z(\beta) = \sum_{i=1}^N e^{-\beta \lambda_i}$. 
\end{proof}

\begin{corollary}\label{diag-decay}
Denote by $\alpha := \frac{1}{N}\sum_{i=1}^N e^{-\beta(\lambda_i-\lambda_1)}$, and let $\mathrm{diag}(\rho)  = \sum_{i=1}^N  \bra{\psi_i}\rho\ket{\psi_i}\ket{\psi_i}\bra{\psi_i}$ be the diagonal part of the state $\rho$ under the basis $(\ket{\psi_i})_{1\le i \le n}$. The dynamics \eqref{lindblad} satisfies 
\begin{equation*}
    \chi^2(\mathrm{diag}(\rho_t), \si) \le e^{-\alpha t}  \chi^2(\mathrm{diag}(\rho_0),\si)\,.
\end{equation*}
\end{corollary}
\begin{proof}
Recalling \eqref{pauli} in Theorem \ref{thm-Davies}, $\mathrm{diag}(\rho_t)$ evolves according to the classical Metropolis-Hastings dynamics. Its convergence is then guaranteed by Theorem \ref{conductance}.
\end{proof}

\begin{proof}[Proof of Theorem \ref{fast-mixing}]
We only need to establish the exponential decay of the $\chi^2$ divergence, $\chi^2(\rho_t, \sigma)$. Observe that  
\begin{align*}
    \chi^2(\rho_t, \sigma) = \left\|\Gamma_{\sigma}^{-1}(\rho_t - \sigma)\right\|_{\sigma,1/2}^2 & = \left\|\Gamma_{\sigma}^{-1}(\mathrm{diag}(\rho_t) - \sigma)\right\|_{\sigma,1/2}^2 + \sum_{1 \leq i \neq j \leq N} \bra{\psi_i} \Gamma_{\sigma}^{-1} \rho_t \ket{\psi_j}^2 \\    
    & = \left\|\Gamma_{\sigma}^{-1}(\mathrm{diag}(\rho_t) - \sigma)\right\|_{\sigma,1/2}^2 + \sum_{1 \leq i \neq j \leq N} \frac{\bra{\psi_i} \rho_t \ket{\psi_j}^2}{\bra{\psi_i} \sigma \ket{\psi_i} \bra{\psi_j} \sigma \ket{\psi_j}}\,.
\end{align*}
By Corollary \ref{diag-decay}, the diagonal terms exhibit exponential decay:  
\begin{align*}
    \left\|\Gamma_{\sigma}^{-1}(\mathrm{diag}(\rho_t) - \sigma)\right\|_{\sigma,1/2}^2 \leq e^{-\alpha t} \left\|\Gamma_{\sigma}^{-1}(\mathrm{diag}(\rho_0) - \sigma)\right\|_{\sigma,1/2}^2\,,
\end{align*}
where the exponent $\alpha$ is given by  
\begin{align*}
    \alpha = \frac{1}{N} \sum_{i=1}^N e^{-\beta(\lambda_i - \lambda_1)}\,.
\end{align*}
For the off-diagonal terms, \eqref{coherencef} from Theorem \ref{thm-Davies} ensures that for any $i, j \in [N]$,  
\begin{align*}
\bra{\psi_i} \rho_t \ket{\psi_j} \leq e^{-\alpha t} \bra{\psi_i} \rho_0 \ket{\psi_j},    
\end{align*}
where we used the inequality  
\begin{align*}
    \sum_k \frac{1}{N} \left( \gamma(\lambda_k - \lambda_i) + \gamma(\lambda_k - \lambda_j) \right) \geq \alpha.
\end{align*}
Combining these results, we conclude that  
\begin{align*}
    \chi^2(\rho_t, \sigma) \leq e^{-\alpha t} \chi^2(\rho_0, \sigma).
\end{align*}
Under the condition \eqref{low-energy}, the exponent $\alpha$ is bounded below as $\alpha = \Omega(r_\beta(H))$. This completes the proof.
\end{proof}
\bibliographystyle{quantum}
\bibliography{apssamp}

\end{document}